\documentclass[preprint]{imsart}
\RequirePackage[OT1]{fontenc}
\RequirePackage{amsthm,amsmath,amsfonts}
\RequirePackage[numbers]{natbib}
\RequirePackage[colorlinks,citecolor=blue,urlcolor=blue]{hyperref}
\usepackage{mathtools}
\usepackage{xcolor}
\usepackage{mathrsfs}

\startlocaldefs

\usepackage{mathrsfs}
\usepackage[margin=1in]{geometry}
\usepackage{enumitem}
\usepackage{booktabs}

\numberwithin{equation}{section}

\theoremstyle{plain}
\newtheorem{theorem}{Theorem}[section]

\newtheorem{lemma}{Lemma}[section]

\newtheorem{proposition}{Proposition}[section]

\theoremstyle{definition}
\newtheorem{remark}{Remark}[section]

\newcommand{\X}{\mathcal{X}}
\newcommand{\mcx}{(X_n)_{n \geq 1}}
\newcommand{\R}{\mathbb{R}}
\newcommand{\numbereqn}{\addtocounter{equation}{1}\tag{\theequation}} 
\renewcommand{\bar}{\overline}

\DeclareMathOperator{\var}{var}
\DeclareMathOperator{\cov}{cov}
\DeclareMathOperator{\N}{N}
\DeclareMathOperator*{\ESS}{ESS}
\DeclareMathOperator*{\BM}{BM} 
\DeclareMathOperator*{\diag}{Diag} 

\DeclarePairedDelimiter\floor{\lfloor}{\rfloor}
\usepackage{lineno}

\allowdisplaybreaks

\endlocaldefs

\begin{document}
    
    \begin{frontmatter}
        
        \title{Estimating accuracy of the MCMC variance estimator: a central limit theorem for  batch means estimators}
        
        \runtitle{CLT for batch means variance estimate}
        
        \begin{aug}
            \author{
                \fnms{Saptarshi} \snm{Chakraborty}\thanksref{m1}, \ead[label=e1]{chakrabs@mskcc.org}
            }
            \author{
                \fnms{Suman K.} \snm{Bhattacharya}\thanksref{m2} \ead[label=e2]{sumankbhattachar@ufl.edu}
            }
            \and
            \author{
                \fnms{Kshitij} \snm{Khare}\thanksref{m2} 
                \ead[label=e3]{kdkhare@stat.ufl.edu}
            }
            
            \runauthor{Chakraborty, Bhattacharya and Khare}

            \affiliation{
                Memorial Sloan-Kettering Cancer Center \thanksmark{m1},
                
                Department of Statistics, University of Florida\thanksmark{m2}   
            }
            
            \address{
                \thanksmark{m1}Department of Epidemiology \& Biostatistics\\
                Memorial Sloan Kettering Cancer Center\\
                485 Lexington Ave\\
                New York, NY 10017, USA \\
                \printead{e1}
            }	
            
            \address{
                \thanksmark{m2}Department of Statistics\\
                University of Florida\\
                101 Griffin Floyd Hall\\
                Gainesville, Florida 32601, USA\\
                \printead{e2} \\
                \printead{e3}
            }
            
        \end{aug}
        
        \begin{abstract}
            The batch means estimator of the MCMC variance is a simple and effective measure of accuracy for MCMC based ergodic averages. Under various regularity conditions, the estimator has been shown to be consistent for the true variance. However, the estimator can be unstable in practice as it depends directly on the raw MCMC output. A measure of accuracy of the batch means estimator itself, ideally in the form of a confidence interval, is therefore desirable.  The asymptotic variance of the batch means estimator is known; however, without any knowledge of asymptotic distribution, asymptotic variances are in general insufficient to describe variability. In this article we prove a central limit theorem for the batch means estimator that allows for the construction of asymptotically accurate confidence intervals for the batch means estimator. Additionally, our results provide a Markov chain analogue of the classical CLT for the sample variance parameter for i.i.d. observations.  Our result assumes standard regularity conditions similar to the ones assumed in the literature for proving consistency. Simulated and real data examples are included as illustrations and applications of the CLT.     
        \end{abstract}
        
        
        
    \end{frontmatter}

    \section{Introduction}
    Markov chain Monte Carlo (MCMC) techniques are indispensable tools of modern day computations. Routinely used in Bayesian analysis and machine learning, a  major application of MCMC lies in the approximation of intractable and often high-dimensional integrals. To elaborate, let $(\X, \mathcal{F}, \nu)$ be an arbitrary measure space and let $\Pi$ be a probability measure on $\X$, with associated density $\pi(\cdot)$ with respect to $\nu$. The quantity of interest is the integral
    \[
    \pi f = E_\pi f :=\int_{\X} f(x) \: d\Pi(x)=\int_{\X} f(x) \:\pi(x) \:\nu(dx)
    \]
    where $f$ is a real-valued, $\Pi-$integrable function on $\X$. In many modern applications, the such an integral is often intractable, i.e., (a) does not have a closed form, (b) deterministic approximations are inefficient, often due to the high dimensionality of $\X$, and (c) cannot be estimated via classical or i.i.d. Monte Carlo techniques as i.i.d. random generation from $\Pi$ is in general infeasible.  Markov chain Monte Carlo (MCMC) techniques are the to-go method of approximation for such integrals. Here, a Markov chain $\mcx$ with an invariant probability distribution $\Pi$ \citep[see, e.g.][for definitions]{meyn:tweedie:2012} is generated using some MCMC sampling technique such as the Gibbs sampler or the Metroplis Hastings algorithms. Then, ergodic averages $\bar f_n := n^{-1} \sum_{i=1}^{n} f(X_i)$ based on realizations of the Markov chain $\mcx$ are used as approximations of $E_\pi f$. 
    
    Measuring the errors incurred in approximations is a critical step in any numerical analysis. It is well known that when a Markov chain is Harris ergodic (i.e., aperiodic, $\phi$-irreducible, and Harris recurrent \citep[see][for definitions]{meyn:tweedie:2012}), then ergodic averages based on realizations of the Markov chain always furnish strongly consistent estimates of the corresponding population quantities \citep[][Theorem 13.0.1]{meyn:tweedie:2012}. In other words, if a Harris ergodic chain is run long enough, then the estimate $\bar f_n$ is \emph{always} guaranteed to provide a reasonable approximation to the otherwise intractable quantity $E_\pi f$ (under some mild regularity conditions on $f$). Determining an MCMC sample (or iteration) size $n$ that justifies this convergence, however, requires a measurement of accuracy. Similar to i.i.d. Monte Carlo estimation, the standard error of $\bar f_n$ obtained from the MCMC central limit theorem (MCMC CLT) is the natural quantity to use for this purpose. MCMC CLT requires additional regularity conditions as compared to its i.i.d. counterpart; if the Markov chain $\mcx$ is geometrically ergodic (see, e.g., \citet{meyn:tweedie:2012} for definitions), and if $E_\pi |f|^{2+\delta}$ for some $\delta > 0$ (or $E_\pi f^2 < \infty$ if $\mcx$ is geometrically ergodic and reversible), it can be shown that as $n \to \infty$
    \[
    \sqrt{n}\left(\bar{f}_n - E_\pi f \right)\xrightarrow{d} \N(0,\sigma^2_f)
    \]
    where $\sigma_f^2$ is the  \emph{MCMC variance} defined as
    \begin{equation} \label{defn_mcmcvar}
    \sigma^2_f = \var_\pi f(X_1) +2\sum_{i=2}^{\infty}\cov_\pi\left(f(X_1),f(X_i)\right).
    \end{equation}
    Here  $\var_\pi$ and $\cov_\pi$ respectively denote the variance and (auto-) covariance computed under the stationary distribution  $\Pi$. Note that other sufficient conditions ensuring the above central limit theorem also exist; see the survey articles of \citet{jones:2004}, and \citet{robert:rosenthal:2004} for more details. When the regularity conditions hold, a natural measure of accuracy for $\bar f_n$ is therefore given by the MCMC standard error (MCMCSE) defined as $\sigma_f/\sqrt{n}$. Note that this formula of MCMCSE, alongside measuring the error in approximation, also helps determine an optimum iteration size $n$ that is required to achieve a pre-specified level of precision, thus providing a stopping rule for terminating MCMC sampling. A related use of $\sigma^2_f$ also lies in the computation of effective sample size $\ESS = n \var_\pi f(X_1)/\sigma^2_f$ \cite{kass:1998, ripley:2009}. ESS measures how $n$ dependent MCMC samples compare to $n$ i.i.d. observations from $\Pi$, thus providing a univariate measure of the quality of the MCMC samples. Thus to summarize, the MCMC variance $\sigma_f^2$ facilitates computation/determination of three crucial aspects of an MCMC implementation, namely  (a) stopping rule for terminating simulation, (b) effective sample size (ESS) of the MCMC draws, and (c)  precision of the MCMC estimate $\bar f_n$. 
    
    In most non-trivial applications, however, the MCMC variance $\sigma^2_f$ is usually unknown, and must be estimated. A substantial literature has been devoted to the estimation of  $\sigma^2_f$  \citep[see, e.g.,][to name a few]{bratley:fox:2011, fishman:2013, geyer:1992, glynn:1990, glynn:whitt:1991, mykland:tierney:yu:1995, roberts:1996, flegal:haran:jones:2008, flegal:jones:2010}, and several methods, such as regerative sampling, spectral variance estimation, and overlapping and non-overlapping batch means estimation, have been developed. In this paper, we focus on the non-overlapping batch means estimator, henceforth called the batch means estimator for simplicity,  where estimation of $\sigma^2_f$ is performed by breaking the $n = a_n b_n$ Markov chain iterations into $a_n$ non-overlapping blocks or batches of equal size $b_n$. Then, for each $k \in \{1,2,\cdots,a_n\}$, one calculates the $k$-th batch mean $\bar{Z}_k := \frac{1}{b_n} \sum_{i=1}^{b_n} Z_{(k-1)b_n+i}$, and the overall mean $\bar{Z} := \frac{1}{a_n} \sum_{i=1}^{a_n} \bar{Z}_k$, where $Z_i=f(X_i)$ for $i=1,2,\dots$, and finally estimates $\sigma^2_f$ by
    \begin{equation} \label{defn_bme}
    \hat{\sigma}^2_{\BM,f} = \hat{\sigma}^2_{\BM,f}(n, a_n, b_n) =  \frac{b_n}{a_n-1}\sum_{k=1}^{a_n}\left(\bar{Z}_k-\bar{Z}\right)^2. 
    \end{equation}
    
    The batch means estimator is straightforward to implement, and can be computed post-hoc without making any changes to the original MCMC algorithm, as opposed to some other methods, such as regeneration sampling. Under various sets of regularity conditions, the batch mean estimator $\hat {\sigma}^2_{\BM, f}$ has been shown to be strongly consistent \citep{damerdji1991strong, hobert2002applicability, jones:2006, flegal:jones:2010} and also mean squared consistent \citep{chien:1997, flegal:jones:2010} for $\sigma_f^2$, provided the batch size $b_n$ and the number of batches $a_n$ both increase with $n$. Note that the estimator depends on the choice of the batch size $b_n$ (and hence the number of batches $a_n = n/b_n$). Optimal selection of the batch-size is still an open problem, and both $b_n = n^{1/2}$ and $b_n = n^{1/3}$ have been deemed desirable in the literature; the former ensures that the batch means $\{\bar Z_k\}$ approach asymptotic normality at the fastest rate (under certain regularity conditions, \citep{chien:1988}), and the latter minimizes the asymptotic mean-squared error of $\hat{\sigma}^2_{\BM,f}$ (under different regularity conditions, \citep{song:schmeiser:1995}). 
    
    It is however important to recognize that consistency alone does not in general justify practical usefulness, and a measurement of accuracy is \emph{always} required to assess the validity of an estimator. It is known that the asymptotic variance of the batch means estimator is given by  $\var \hat{\sigma}^2_{\BM, f} = 2 \sigma_f^4/a_n + o(1/n)$, under various regularity conditions \citep{chien:1997, flegal:jones:2010}. However, without any knowledge of the asymptotic distribution, the asymptotic variance alone is generally insufficient for assessing the accuracy of an estimator. For example, a $\pm 2$ standard error bound does not in general guarantee more than $75\%$ coverage as obtained from the Chebyshev inequality, and to ensure a pre-specified ($95\%$) coverage, a much larger interval ($\sim \pm 4.5$ standard error) is necessary in general. This provides a strong practical motivation for determining the asymptotic distribution of the batch means estimator. To the best of our knowledge, however, no such result is available.  
    
    The main purpose of this paper is to establish a central limit theorem that guarantees asymptotic normality of the batch means estimator under mild and standard regularity conditions (Theorem~\ref{thm_clt_bme}). There are two major motivations for our work. As discussed above, the first motivation lies in the immediate practical implication of this work. As a consequence of the CLT, the use of approximate normal confidence intervals for measuring accuracy of batch means estimators is justified. Given MCMC samples, such intervals can be computed alongside the batch means estimator at virtually no additional cost, and  therefore could be of great practical relevance. The second major motivation comes from a theoretical point of view. Although a central limit theorem for the sample variance of an i.i.d. Monte Carlo estimate is known (can be easily established via delta method, for example), no Markov chain Monte Carlo analogue of this result is available. Our paper provides an answer to this yet-to-be-addressed theoretical question. The proof is quite involved and leverages  operator theory and the martingale central limit theorem \citep[see, e.g.,][]{alj:azrak:melard:2014}, as opposed to the Brownian motion based approach adopted in \citep{flegal:jones:2010}, and the result is analogous to the classical CLT for sample variance in the i.i.d. Monte Carlo case.

    The remainder of this article is organized as follows. In Section~\ref{sec_clt} we state and prove the main central limit theorem along with a few intermediate results. Section~\ref{sec_illus} provides two illustrations of the CLT -- one based on a toy example (Section~\ref{sec_illus_toy}), and one based on a real world example (Section~\ref{sec_illus_blasso}). Proofs of some key propositions and intermediate results are provided in the Appendix.

    \section{A Central Limit Theorem for Batch-Means Estimator} \label{sec_clt}
    
    This section provides our main result, namely, a central theorem for the non-overlapping batch-means standard error estimator. Before stating the theorem, we fix our notations, and review some known results on Markov chains. Let $\mcx$ be a Markov chain on $(\X, \mathcal{F}, \nu)$ with Markov transition density $k(\cdot, \cdot)$, and stationary measure $\Pi$ (with density $\pi$). We denote by $K(\cdot, \cdot)$, the Markov transition function of $\mcx$; in particular, for $x \in \X$ and a Borel set $A \subseteq \X$, $K(x, A) = \int_{A} k(x, x') \: dx'$. For $m \geq 1$, the associated $m$-step Markov transition function is defined in the following inductive fashion
    \[
    K^m(x, A) = \int_{\R^p} K^{m-1}(x', A) \: K(x, dx') = \text{Pr}(X_{m+j} \in A \mid X_j = x)
    \]
    for any $j = 0, 1, \dots$, with $K^1 \equiv K$. The Markov chain $\mcx$ is said to be reversible, if for any $x, x' \in \X$ the \textit{detailed balance condition}
    \[
    \pi(x) K(x, dx') = \pi(x') K(x', dx)
    \]
    is satisfied. Also, the chain $\mcx$  is said to be geometrically ergodic if there exists a constant $\kappa \in [0, 1)$ and a function $Q : \X \to [0, \infty)$ such that for any $x \in \X$ and any $m \in \{1, 2, \dots \}$
    \[
    \|K^m(x, \cdot) - \Pi(\cdot) \| := \sup_{A \subseteq \X} |K^m(x, A) - \Pi(A) | \leq Q(x) \kappa^m.
    \]
    
    \noindent Let us denote by
    \[
    L^2_0(\pi) = \left\{ g: \X \to \R : E_\pi g = \int_\X f(x) \:  d\Pi(x) = 0 \text{ and } 
    E_\pi g^2 = \int_\X g(x)^2 \:  d\Pi(x) < \infty \right\}.
    \]
    This is a Hilbert space where the inner product of $g,h \in L^2_0(\pi)$ is defined as
    \[
    \langle g,h\rangle_\pi = \int_{\X} g(x) \: h(x) \:d\Pi(x) = \int_{\X} g(x) \: h(x) \: \pi(x) \: d\nu(x) 
    \]
    and the corresponding norm is defined by $\|g\|_{\pi} = \sqrt{\langle g, g\rangle}_\pi$. The Markov transition function $K(\cdot, \cdot)$ determines a Markov operator; we shall slightly abuse our notation and denote the associated operator by $K$ as well. More specifically, we shall let $K: L^2_0(\pi) \to L^2_0(\pi)$ denote the operator that maps $g \in L^2_0(\pi)$ to $(Kg)(x) = \int_{\X} g(x') K(x, dx')$. The operator norm of $K$ is defined as $\|K \| = \sup_{g \in L^2_0(\pi) : \|g\|_\pi = 1} \|Kg\|$. It follows that $\| K \| \leq 1$.  \citet{roberts:rosenthal:1997} show that for reversible (self-adjoint) $K$,
    $\| K \| < 1$ if and only if the associated Markov chain $\mcx$ is geometrically ergodic.
    
    The following theorem establishes a CLT for the batch means estimator of MCMC variance.
    
    \begin{theorem} \label{thm_clt_bme}
        Suppose ${(X_n)}_{n \geq 1}$ is a stationary geometrically ergodic reversible Markov chain with state space $\X$ and invariant distribution $\Pi$. Let $f: \X \to \R$ be a Borel function with $E_\pi (f^8) > 0$.  Consider the batch means estimator $\hat \sigma^2_{\BM, f} = \hat \sigma^2_{\BM, f} (n, a_n, b_n)$ of the MCMC variance $\sigma^2_{f}$ as defined in \eqref{defn_bme}. Let $a_n$ and $b_n$ be such that  $a_n \to \infty$, $b_n \to \infty$ and $\sqrt{a_n}/b_n \to 0$ as $n \to \infty$. Then 
        \[
        \sqrt{a_n} \left(\hat{\sigma}^2_{\BM,f} (n, a_n, b_n)-\sigma^2_f\right) \xrightarrow{d} \N \left(0, 2\sigma^4_f \right)
        \]
        where $\sigma^2_f$ is the MCMC variance as defined in \eqref{defn_mcmcvar}. 
    \end{theorem}
    
    \begin{remark}[\textbf{Proof technique}] \label{remark_proof_technique}
    	Our proof is based on an operator theoretic approach, and relies on a careful manipulation of appropriate moments, and the martinagle CLT. Previous work in \citep{chien:1997, damerdji1991strong, flegal:jones:2010} on consistency of $\hat{\sigma}^2_{\BM, f}$ is based on a Brownian motion based approximation (see \cite[Equation~??]{flegal:jones:2010}). This leads to some differences in the assumptions that are required to prove the respective results. Note again that \citep{chien:1997, damerdji1991strong, flegal:jones:2010} do not explore a CLT for the batch means estimator. 
    \end{remark}

    \begin{remark}[\textbf{Discussion of assumptions: Uniform vs. Geometric ergodicity, reversibility and moments}] \label{remark_uniform_ergo}
    	Our results require geometric ergodicity of the Markov chain, which in general is required to guarantee CLT of the MCMC estimate $\bar f_n$ itself.  The consistency of $\hat \sigma_{\BM, f}^2$ in \cite{chien:1997} and \citep{damerdji1991strong} have been proved under uniform ergodicity of the Markov chain, which is substantially more restrictive and difficult to justify in practice.  On the other hand, \cite{flegal:jones:2010} consider a Brownian motion based approach to prove their result. The consistency result in \cite{flegal:jones:2010} holds under geometric ergodicity, however, verifying a crucial Brownian motion based sufficient condition can be challenging when the chain is not uniformly ergodic.  
    	
    	On the other hand, we require reversibility of the  Markov chain which is not a requirement in \citep{chien:1997, damerdji1991strong, flegal:jones:2010}. Note that the commonly used Metropolis-Hastings algorithm and its modern efficient extension, the Hamiltonian Monte Carlo algorithm, are necessarily reversible \citep{geyer:1992, neal2011mcmc}. Also, for any Gibbs sampler, a reversible counterpart can always be constructed through random scans or reversible fixed scans \citep{besag1986statistical, geyer:1992}, and a two-block Gibbs sampler is \emph{always} reversible. 
    	
    	We require the function $f$ to have a finite eighth moment, while the consistency results in \cite{damerdji1991strong} assume the existence of twelfth moment and those in \cite{flegal:jones:2010} assume moments of order $4 + \delta + \epsilon$ for some $\delta > 0$ and $\epsilon > 0$. Note again that the authors in \citep{flegal:jones:2010} do not establish a CLT. 
    \end{remark}

    \begin{remark}[\textbf{Stationarity}]  \label{remark_burnin}
        It is to be noted that Theorem~\ref{thm_clt_bme} assumes stationarity, i.e., the initial measure of the Markov chain is assumed to be the stationary measure. This is similar to the assumptions made in \citep{damerdji1991strong, chien:1988} for establishing consistency. A moderate burn-in or warm-up period for an MCMC algorithm is usually enough to guarantee stationarity in practice.
    \end{remark}

	\begin{remark}[\textbf{Choice of $a_n$ and $b_n$}] \label{remark_an_bn_choice}
		Consider the two practically recommended choices \citep{flegal:haran:jones:2008} (i) $a_n = b_n = \sqrt{n}$ and (ii) $\sqrt{a_n} = b_n = n^{1/3}$ as mentioned in the Introduction. Clearly, (i) satisfies the sufficient conditions on $a_n$ and $b_n$ described in Theorem~\ref{thm_clt_bme} and hence, batch means estimators based on this choice attains a CLT, provided the other conditions in Theorem~\ref{thm_clt_bme} hold. On the other hand, (ii) does not satisfy the conditions in Theorem~\ref{thm_clt_bme}, and hence a CLT is not guaranteed with this choice. Small adjustments, such as $a_n = n^{ -\delta + 2/3}, b_n = n^{\delta + 1/3}$ for some small $0 < \delta < 2/3$, and $a_n = n^{2/3} (\log n)^{-\delta}$ and $b_n = n^{1/3} (\log n)^{\delta}$ for some (small) $\delta > 0$, could be used to technically satisfy the sufficient condition, however, the resulting convergence in distribution may be slow (see the toy example in Section~\ref{sec_illus_toy}). 
	\end{remark}

    Before proving Theorem~\ref{thm_clt_bme}, we first introduce some notation, and then state and prove some intermediate results. Suppose the Markov chain $(X_n)_{n \geq 1}$ and the function $f$ satisfy the assumptions made in Theorem~\ref{thm_clt_bme}. Define $Y_i = f(X_i) - E_\pi f$ for $i = 1, 2, \dots$, and write the batch-means estimator $\hat\sigma_{\BM, f}^2$ in \eqref{defn_bme} as
    \begin{equation*}
    \hat{\sigma}^2_{\BM,f} = \hat{\sigma}^2_{\BM,f}(n, a_n, b_n) =  \frac{b_n}{a_n-1}\sum_{k=1}^{a_n}\left(\bar{Y}_k-\bar{Y}\right)^2 =\frac{a_n}{a_n-1} \left(\frac{b_n}{a_n}\sum_{k=1}^{a_n}{\bar{Y}_k}^2-b_n\bar{Y}^2 \right).
    \end{equation*}
    Here $\bar{Y}_k := b_n^{-1} \sum_{i=1}^{b_n} Y_{(k-1)b_n+i}$, and $\bar{Y} := a_n^{-1} \sum_{i=1}^{a_n} \bar{Y}_k$. We shall consider the related quantity 
    \begin{equation}\label{defn_mod_bme}
    \tilde{\sigma}^2_{\BM,f} := \frac{b_n}{a_n}\sum_{k=1}^{a_n}{\bar{Y}_k}^2 - b_n\bar{Y}^2 = \left(\frac{a_n-1}{a_n}\right) \hat \sigma_{\BM, f}^2
    \end{equation}
    and call it the \emph{modified batch means estimator}. The following two lemmas establish two asymptotic results on the modified batch means estimator. The first lemma proves asymptotic normality for the modified batch means estimator (with a shift) whenever $a_n 
\to \infty$ and $b_n \to \infty$. Key propositions needed in the proof of this lemma are provided in the Appendix.

    \begin{lemma} \label{LEMMA_2}
    	Consider the modified batch means estimator $\tilde \sigma^2_{\BM, f}$ as defined in \eqref{defn_mod_bme}. If $a_n \to \infty$ and $b_n \to \infty$ as $n \to \infty$, then 
    	\[
    	\sqrt{a_n} \left[\left(\tilde{\sigma}^2_{\BM, f}  - \sigma_f^2\right) - \left(E_\pi(\tilde{\sigma}^2_{\BM, f}) - E_\pi \left(b_n \bar{Y}^2\right) - \sigma_f^2\right) \right] \longrightarrow \N(0,2\sigma^4_f)
    	\]
    	where $\sigma_f^2$ is the MCMC variance as defined in \eqref{defn_mcmcvar}.
    \end{lemma}
    
    \begin{proof}
    	First observe that
    	\begin{align*} \label{simplify_lemma2_1}
    	& \quad \sqrt{a_n}  \left[\left(\tilde{\sigma}^2_{\BM, f}  - \sigma_f^2\right) - \left(E_\pi(\tilde{\sigma}^2_{\BM, f}) - E_\pi \left(b_n \bar{Y}^2\right) - \sigma_f^2\right) \right] \\
    	&= \sqrt{a_n} \left(\tilde{\sigma}^2_{\BM, f} - E_\pi(\tilde{\sigma}^2_{\BM, f}) - E_\pi \left(b_n \bar{Y}^2\right) \right) \\
    	&= \sqrt{a_n} \left(\frac{b_n}{a_n} \sum_{k=1}^{a_n} \bar{Y}_k^2 - b_n \bar{Y}^2 - \frac{b_n}{a_n} \sum_{k=1}^{a_n} E_\pi\left(\bar{Y}_k^2\right) \right)\\
    	&= \frac{b_n}{\sqrt{a_n}} \sum_{k=1}^{a_n} \left(\bar{Y}_k^2 - E_\pi \left(\bar{Y}_k^2\right) \right) - \sqrt{a_n} \ b_n \bar{Y}^2\\
    	&= \frac{b_n}{a_n} \sum_{k=2}^{a_n} \left\{\bar{Y}_k^2 - E\left(\bar{Y}_k^2 \mid \mathcal{F}_{k-1}\right) + E\left(\bar{Y}_k^2 \mid \mathcal{F}_{k-1}\right) - E_\pi\left(\bar{Y}_k^2\right) \right\} \\ 
    	& \qquad \qquad + \frac{b_n}{\sqrt{a_n}} \left(\bar{Y_1}^2 - E_\pi\left(\bar{Y_1}^2\right)\right) - \sqrt{a_n}\ b_n\bar{Y}^2\\
    	&= \frac{b_n}{a_n} \sum_{k=2}^{a_n} \left\{\bar{Y}_k^2 - E\left(\bar{Y}_k^2 \mid \mathcal{F}_{k-1}\right) + h\left(X_{(k-1)b_n}\right) - E_\pi h\left(X_{(k-1)b_n}\right) \right\} \\
    	& \qquad \qquad + \frac{b_n}{\sqrt{a_n}} \left(\bar{Y_1}^2-E_\pi\left(\bar{Y_1}^2\right) \right) - \sqrt{a_n}\ b_n \bar{Y}^2. \numbereqn
    	\end{align*}
    	Here, for  $1 \leq k \leq a_n$, 
    	$\mathcal{F}_{k, n}$ is the sigma-algebra generated by $X_1, \dots, X_{kb_n}$, and
    	\[
    	h(X_{(k-1)b_n}) := E\left(\bar{Y}_k^2 \mid \mathcal{F}_{k, n} \right) = E \left(\bar{Y}_k^2 \mid X_{(k-1)b_n} \right)
    	\] 
    	due to the Markovian structure of $\mcx$. Let $\tilde{h} = h - E_\pi h \in \mathcal{L}_0^2(\pi)$. Since the Markov operator $K$ has operator norm $\lambda = \|K\| < 1$ (due to geometric ergodicity), it follows  that $I - K$ is invertible (using, e.g., the expansion $(I - K)^{-1} = \sum_{j=0}^\infty K^j$). Therefore, $I - K^{b_n}$ is also invertible, since $K^{b_n}$ is also a Markov operator. Consequently, one can find a $\tilde g$ such that $\tilde g = (I- K^{b_n})^{-1} \tilde h$, i.e., $\tilde{h} = \tilde{g} - K^{b_n} \tilde{g}$. Then
    	\begin{align*}
    	h\left(X_{(k-1)b_n}\right) - E_\pi h(X_{(k-1)b_n}) 
    	&= \tilde{h} \left(X_{(k-1)b_n}\right) \\
    	&= \tilde{g}\left(X_{(k-1)b_n}\right) - K^{b_n} \tilde{g}\left(X_{(k-1)b_n}\right)\\
    	&= \left[ \tilde{g}\left(X_{(k-1)b_n}\right) - \tilde{g}\left(X_{kb_n}\right)\right] + \left[\tilde{g}\left(X_{kb_n}\right) - K^{b_n} \tilde{g}\left(X_{(k-1)b_n}\right)\right]
    	\end{align*}
    	
    	\noindent Hence 
    	\begin{align*}
    	&\quad \sum_{k=2}^{a_n} \left[ h\left(X_{(k-1)b_n}\right) - E_\pi h\left(X_{(k-1)b_n}\right) \right] \\
    	&= \sum_{k=2}^{a_n} \left[ \tilde{g} \left(X_{kb_n}\right) - E\left(\tilde{g}(X_{kb_n}) \mid \left(X_{(k-1)b_n}\right) \right) \right] + \tilde{g}\left(X_{b_n}\right) - \tilde{g}\left(X_{a_nb_n}\right)
    	\end{align*}
    	so that from \eqref{simplify_lemma2_1}, 
    	\begin{align*} \label{simplify_lemma2_2}
    	&\quad\sqrt{a_n} \left(\tilde{\sigma}^2_{\BM, f} - E_\pi(\tilde{\sigma}^2_{\BM, f}) - E_\pi \left(b_n \bar{Y}^2\right) \right) \\
    	&= \frac{b_n}{\sqrt{a_n}}\sum_{k=2}^{a_n}\left[\bar{Y}_k^2 - E\left(\bar{Y}_k^2 \mid X_{(k-1)b_n}\right) + \tilde{g}\left(X_{kb_n}\right) - E\left(\tilde{g}(X_{kb_n}) \mid X_{(k-1)b_n}\right) \right] \\
    	& \qquad + \frac{b_n}{\sqrt{a_n}} \left(\tilde{g}\left(X_{kb_n}\right) - \tilde{g}\left(X_{a_n b_n}\right) \right)  + \frac{b_n}{\sqrt{a_n}}\left(\bar{Y}_1^2 - E_\pi\left(\bar{Y}_1^2\right)\right) - \sqrt{a_n}\ b_n\bar{Y}^2 \\
    	&= T_1 + T_2 + T_3 - T_4, \text{ say.} \numbereqn
    	\end{align*}
    	
    	\noindent We shall note the convergences of the terms $T_1$, $T_2$, $T_3$ and $T_4$ separately. From Markov chain CLT, we have $\sqrt{n}\ \bar{Y} = \sqrt{a_n b_n}\ \bar{Y} \xrightarrow{d} \N(0, \sigma_f^2)$.
    	Therefore, $(\sqrt{n}\ \bar{Y})^2 = a_n b_n \bar{Y}^2 = \mathcal{O}_P(1)$, which means, 
    	\[
    	T_4 = \sqrt{a_n}\ b_n \bar{Y}^2 = \frac{1}{\sqrt{a_n}} \cdot a_n b_n \bar{Y}^2 = o_P(1).
    	\]
    	Again, for all $1 \leq k \leq a_n$, 
    	\begin{align*}
    	\left\| b_n\  \tilde{g}(X_{kb_n})\right\|^2_\pi &= b_n^2\ \left\|(I- K^{b_n})^{-1} \tilde{h}(X_{kb_n}) \right\|_\pi^2 \\
    	&= b_n^2 \left\|\left(\sum_{j=0}^\infty K^{b_n j}\right) \tilde{h}(X_{kb_n})\right\|^2_\pi \\
    	&\leq b_n^2 \left(\sum_{j=0}^\infty \|K\|^{b_n j}\right)^2 \left\|\tilde{h}(X_{kb_n})\right\|_\pi^2 \\
    	&= \left(\frac{1}{1-\lambda^{b_n}}\right)^{2} \var_\pi E \left(b_n\ \bar{Y}_k^2 \mid X_{(k-1)b_n} \right)   \\
    	&\leq \left(\frac{1}{1-\lambda^{b_n}}\right)^{2}  E_\pi \left(b_n^2\ \bar{Y}_k^4  \right) \to 3  \sigma_f^4 
    	\end{align*}
    	since $\sum_{j=0}^\infty \|K\|^{b_n j} = (1 - \lambda^{b_n})^{-1} \to 1$ as $\lambda = \|K\| \in (0, 1)$ and $E(b_n^2\ \bar{Y}_k^4) \to 3 \sigma_f^4$ from Proposition~\ref{prop_Y1bar4}.   Consequently, $b_n \tilde g(X_{kb_n}) = \mathcal{O}_p(1)$ and hence
    	\[
    	T_2 = \frac{b_n}{\sqrt{a_n}} \left(\tilde{g}\left(X_{kb_n}\right) - \tilde{g}\left(X_{a_n b_n}\right) \right) = o_P(1).
    	\]  
    	Again using the Markov chain CLT for $\bar Y_1$, it follows that 
    	\[
    	T_3 = \frac{b_n}{\sqrt{a_n}}\left(\bar{Y}_1^2 - E_\pi\left(\bar{Y}_1^2\right)\right) = o_P(1).
    	\]
    	
    	\noindent Finally, note that the terms inside the summation sign in $T_1$, i.e.,
    	\[
    	\zeta_{k,n} = \bar{Y}_k^2 - E\left(\bar{Y}_k^2 \mid X_{(k-1)b_n}\right) + \tilde{g}\left(X_{kb_n}\right) - E\left(\tilde{g}(X_{kb_n}) \mid X_{(k-1)b_n}\right)
    	\] 
    	forms a martingale difference sequence (MDS), for $k \geq 2$. Let 
    	\begin{equation} \label{xi_defn}
    	\xi_{k, n} = \zeta_{k,n}/\sqrt{E_\pi (\zeta_{k,n}^2)}
    	\end{equation} 
    	Of course $E_\pi \xi_{k-1, n} = 0$, $\var_\pi \xi_{k, n} = 1$ and  $E_\pi\left|\xi_{k,n}\right|^{2+\delta} < \infty$, e.g., for $\delta = 1$ as $E_\pi (f^{8}) < \infty$, by assumption. Then,  for each $n\geq 1$,  $\left(\xi_{k, n}\right)_{k\geq 2}$ is a  mean $0$ and variance $1$ MDS with $(a_n - 1)^{-1} \sum_{k=2}^{a_n} E (\xi_{k,n}^2 \mid \mathcal{F}_{k, n}) \xrightarrow{P} 0$ (Proposition~\ref{prop_Exi2givenF} in Appendix~\ref{sec_appen_proof}). Therefore,
    	\[
    	\frac{1}{\sqrt{a_n-1}} \sum_{k=2}^{a_n}\xi_{k,n} \xrightarrow{d} \N(0,1)
    	\]
    	as $n\to\infty$, by the Lyapunov CLT for MDS \citep[Theorem~1.3]{alj:azrak:melard:2014}. Hence,
    	\[
    	T_1 = \frac{b_n}{\sqrt{a_n}} \sum_{k=2}^{a_n} \zeta_{k,n} =  \frac{b_n}{\sqrt{a_n}} \sum_{k=2}^{a_n} \tau_n \xi_{k,n} = b_n\tau_n \frac{1}{\sqrt{a_n}} \sum_{k=2}^{a_n} \xi_{k,n} \xrightarrow{d} \N(0,c^2)
    	\]  
    	as long as $b_n^2 \tau_n^2 \to c^2$ as $n\to\infty$ for some $c>0$, where $\tau_n^2 = E_\pi(\zeta_{k,n}^2)$. 
    	Now,
    	\[
    	b_n^2\tau_n^2 = E_\pi 
    	\left[ b_n \bar{Y}_1^2 - E\left(b_n\bar{Y}_1^2 \mid X_{0}\right) + b_n \tilde{g}\left(X_{b_n}\right) -  E\left(b_n \tilde{g}(X_{b_n}) \mid X_{0}\right) \right]^2=E_\pi\left[U_n+V_n\right]^2
    	\]
    	where 
    	\begin{equation} \label{defn_U_n}
    	U_n = b_n \bar{Y}_1^2 - E\left(b_n\bar{Y}_1^2 \mid X_{0}\right) + b_n\tilde{h}\left(X_{b_n}\right)
    	\end{equation} 
    	and
    	\begin{equation} \label{defn_V_n} 
    	V_n = b_n \tilde{g}\left(X_{b_n}\right) - b_n\tilde{h}\left(X_{b_n}\right) -  E\left(b_n \tilde{g}(X_{b_n}) \mid X_{0}\right). 
    	\end{equation}
    	From Propositions~\ref{prop_U_n} and \ref{prop_V_n} in Appendix~\ref{sec_appen_proof}, it follows that 
    	$E_\pi (U_n^2) \to 2\sigma_f^4$ and $E_\pi (V_n^2) \to 0$ as $n \to \infty$, where $\sigma_f^2$ is the MCMC variance \eqref{defn_mcmcvar}. Therefore, by Schwarz's inequality, $0 \leq \{E_\pi (U_n V_n)\}^2 \leq E_\pi (U_n^2) E_\pi(V_n^2) \to 0$, i.e., $E(U_n V_n) \to 0$ and hence  
    	\[
    	b_n^2 \tau_n^2 = E_\pi(U_n + V_n)^2 =  E_\pi (U_n^2 + V_n^2 + 2 U_n V_n) \to 2\sigma_f^4.
    	\] 
    	Consequently, $T_1 \xrightarrow{d} \N(0, 2\sigma_f^4)$. Using this in \eqref{simplify_lemma2_2}, together with the fact that each of $T_2$, $T_3$ and $T_4$ is $o_P(1)$, completes the proof.
    \end{proof}
    
     We now state and prove our second lemma. This lemma shows that the shift in  Lemma~\ref{LEMMA_2} is asymptotically negligible if  $a_n$ is of an order smaller  than  $n^{1/3}$. On the other hand, if $a_n$ is of a larger order than $n^{1/3}$, and $K$ is a positive operator ($\langle g, Kg \rangle_\pi \geq 0$ for all $g \in L_0^2 (\pi)$), then the shift diverges to infinity asymptotically.

    \begin{lemma} \label{LEMMA_1}
        Consider the modified batch means estimator $\tilde \sigma^2_{\BM, f}$ as defined in \eqref{defn_mod_bme}. 
        As $n \to \infty$, we have, 
        \begin{enumerate} [label=(\roman*)]
            \item \label{lemma1_zero} $\displaystyle \sqrt{a_n} \left| E_\pi\left(\tilde{\sigma}^2_{\BM,f}\right) + E_\pi\left(b_n\bar{Y}^2\right) - \sigma^2_f \right| \to 0$ if $\sqrt{a_n}/b_n \to 0$, 
            \item \label{lemma1_infty} in addition, if the Markov operator $K$ associated with $\mcx$ is positive, self-adjoint, and $K(f - E_\pi f) \not\equiv 0$, then  $\displaystyle \sqrt{a_n} \left| E_\pi\left(\tilde{\sigma}^2_{\BM,f}\right) + E_\pi\left(b_n\bar{Y}^2\right) - \sigma^2_f \right| \to \infty$ if $\sqrt{a_n}/b_n \to \infty$.
        \end{enumerate}
    \end{lemma}
    
    \begin{proof}
        
        On the outset, note that	
        \begin{align*}
        \sqrt{a_n} \left(E_\pi\left(\tilde{\sigma}^2_{\BM,f}\right) + E_\pi\left(b_n\bar{Y}^2\right) - \sigma^2_f\right)
        &= \sqrt{a_n}\left(E_\pi\left[\frac{b_n}{a_n}\sum_{k=1}^n\bar{Y}_k^2\right]-\sigma^2_f\right) \\
        &= \sqrt{a_n}\left(b_n E_\pi \left(\bar{Y_1}^2 \right)-\sigma^2_f\right). \numbereqn \label{simplify_lemma1_1}
        \end{align*}
        where $\sigma_f^2$ is  the MCMC variance defined in \eqref{defn_mcmcvar}. Now
        \[
        b_n E_\pi \left(\bar{Y_1}^2 \right) = \frac{1}{b_n} \: E_\pi(Y_1 + Y_2 + \dots + Y_{b_n})^2 = \frac{1}{b_n} \left(b_n\gamma_0 + 2\sum_{k=1}^{b_n-1}(b_n - k) \gamma_k\right)
        \]
        and from \eqref{defn_mcmcvar}, $\sigma_f^2 = \gamma_0 + 2 \sum_{k=1}^\infty \gamma_k$ where for any $h \geq 0$, $\gamma_h$ denotes the auto-covariance
        \begin{equation} \label{defn_gamma_h}
        \gamma_h = \cov_\pi(Y_1,Y_{1+h}) = E_\pi(Y_1 Y_{1+h}) =  E_\pi [Y_1 E(Y_{1+h} \mid X_1)] = \langle f_0, K^{h} f_0 \rangle.
        \end{equation}
        Here $f_0 = f - E_\pi f \in L^2_0(\pi)$, $K^0 \equiv I$ (the identity operator), and $K^h$ for $h \geq 1$ denotes the operator associated with the $h$-step Markov transition function. Therefore, from \eqref{simplify_lemma1_1}, it follows that
        \begin{align*} \label{simplify_lemma1_2}
        &\quad \sqrt{a_n}\left| E_\pi\left(\tilde{\sigma}^2_{\BM,f}\right) + E_\pi\left(b_n\bar{Y}^2\right) - \sigma^2_f \right| \\ 
        &= \sqrt{a_n}\left| \frac{1}{b_n} \left(b_n\gamma_0 + 2\sum_{k=1}^{b_n-1}(b_n - k) \gamma_k\right) -
        \left(\gamma_0 + 2 \sum_{k=1}^\infty \gamma_k\right) \right| \\
        &=  \frac{\sqrt{a_n}}{b_n}  \left| b_n\gamma_0 + 2\sum_{k=1}^{b_n-1}(b_n - k) \gamma_k  -
        b_n\gamma_0 - 2 b_n \sum_{k=1}^\infty \gamma_k \right| \\
        &= \frac{\sqrt{a_n}}{b_n} \left|- 2 \sum_{k=1}^{b_n-1} k \gamma_k - 2 b_n \sum_{k=b_n}^\infty \gamma_k \right|
        = \frac{2\sqrt{a_n}}{b_n} \left|\sum_{k=1}^{b_n-1} k \gamma_k + b_n \sum_{k=b_n}^\infty \gamma_k \right| \numbereqn.
        \end{align*}
        
        \noindent Using triangle inequality on the right hand side of \eqref{simplify_lemma1_2}, we get
        \begin{align*} 
        \sqrt{a_n}\left| E_\pi\left(\tilde{\sigma}^2_{\BM,f}\right) + E_\pi\left(b_n\bar{Y}^2\right) - \sigma^2_f \right| 
        &\leq \frac{2\sqrt{a_n}}{b_n} \left(\sum_{k=1}^{b_n-1} k |\gamma_k| + b_n \sum_{k=b_n}^\infty |\gamma_k| \right) \\
        & \stackrel{(\star)}{\leq} \frac{2\sqrt{a_n}}{b_n} \|f_0\|_\pi^2 \left(\sum_{k=1}^{b_n-1} k \lambda^k + b_n \sum_{k=b_n}^\infty \lambda^k \right) \\
        &\leq \frac{2\sqrt{a_n}}{b_n} \|f_0\|_\pi^2 \sum_{k=1}^\infty k \lambda^k = \frac{2\sqrt{a_n}}{b_n} \cdot \frac{\lambda}{1- \lambda}.
         \numbereqn \label{bias_term_upper}
        \end{align*}
        It follows that $\sqrt{a_n} | E_\pi(\tilde{\sigma}^2_{\BM,f}) + E_\pi (b_n\bar{Y}^2) - \sigma^2_f| \to 0$ if $\sqrt{a_n}/b_n \to 0$ as $n \to \infty$.  Here $\lambda = \|K\| < 1$ (as the chain is geometrically ergodic), and $(\star)$ follows from the fact that $|\gamma_h| = | \langle f_0, K^h f_0 \rangle_\pi | \leq \|K\|^h \|f_0\|_\pi^2 = \lambda^h \|f_0\|_\pi^2$. This proves \ref{lemma1_zero}.  
        
        As for \ref{lemma1_infty}, note that if $K$ is a positive operator, then $\gamma_h = \langle f_0, K^h f_0 \rangle_\pi \geq 0$ for all $h \geq 0$. Moreover, reversibility of $\mcx$ implies, $\gamma_2 = \langle f_0, K^2 f_0 \rangle_\pi = \langle K f_0, K f_0 \rangle_\pi = \|Kf_0\|_\pi^2 > 0$ (since $Kf_0 \not \equiv 0$ by assumption). Consequently, the terms under the absolute sign in the right hand side of \eqref{simplify_lemma1_2} is bounded below by $2\gamma_2 > 0$. As such
        \begin{equation}
        \sqrt{a_n}\left|E_\pi\left(\tilde{\sigma}^2_{\BM,f}\right) + E_\pi\left(b_n\bar{Y}^2\right) - \sigma^2_f \right| \geq 4\frac{\sqrt{a_n}}{b_n} \gamma_2.
        \label{bias_term_lower}
        \end{equation}
        It follows that $\sqrt{a_n} | E_\pi(\tilde{\sigma}^2_{\BM,f}) + E_\pi (b_n\bar{Y}^2) - \sigma^2_f| \to \infty$ if  $\sqrt{a_n}/b_n \rightarrow \infty$ as $n \rightarrow \infty$. This proves (ii). This proves \ref{lemma1_infty}.
    \end{proof}

    With Lemma~\ref{LEMMA_1} and \ref{LEMMA_2} proved, we are now finally in a position to formally prove Theorem~\ref{thm_clt_bme}, which is essentially a combination of these two lemmas, and the fact that the modified batch means estimator is asymptotically equivalent to the batch means estimator.
    
    
    \begin{proof}[Proof of Theorem~\ref{thm_clt_bme}]
        Observe that 
        \begin{align*}
        \sqrt{a_n} \left(\tilde{\sigma}^2_{\BM,f} - \sigma^2_f\right) &= \sqrt{a_n} \left(E_\pi \left(\tilde{\sigma}^2_{\BM,f}\right) + E_\pi\left(b_n\bar{Y}^2\right) - \sigma^2_f \right) \\ 
        &\qquad + \sqrt{a_n} \left( \tilde{\sigma}^2_{\BM, f} - E_\pi(\tilde{\sigma}^2_{\BM, f}) - E_\pi \left(b_n \bar{Y}^2\right) \right)  \\
        &\xrightarrow{d} \N \left(0, 2 \sigma_f^4\right),
        \end{align*}
        from Lemma~\ref{LEMMA_1}, Lemma~\ref{LEMMA_2} and Slutsky's theorem. Therefore,
        \begin{align*}
        \sqrt{a_n} \left(\hat{\sigma}^2_{\BM,f} - \sigma^2_f\right) &= \sqrt{a_n} \left[ \left(\frac{a_n}{a_n-1}\right) \tilde{\sigma}^2_{\BM, f} - \sigma^2_f \right] \\
        &= \left(\frac{a_n}{a_n-1}\right) \sqrt{a_n} \left(\tilde{\sigma}^2_{\BM,f} - \sigma^2_f\right) - \left(\frac{\sqrt{a_n}}{a_n - 1}\right) \sigma_f^2 \\
        &\xrightarrow{d} \N \left(0, 2 \sigma_f^4\right),
        \end{align*}
        by another application of Slutsky's theorem. This completes the proof.
    \end{proof}

    \section{Illustration} \label{sec_illus}
    This section illustrates the applicability of the central limit theorem through replicated frequentist evaluations of the batch means MCMC variance estimator. To elaborate, given a total  iteration size $n+n_0$, where $n$ denotes the final MCMC iteration size and $n_0$ denotes the burn-in size, we generate  replicated $(n+n_0)$-realizations of a Markov chain  with different and independent random starting points, and evaluate an appropriate function $f$ at each Markov chain realization. The batch means MCMC variance estimates $\hat \sigma_{\BM, f}^2(n, a_n, b_n)$ for a few different choices of $b_n$ (and $a_n = n/b_n$) are subsequently computed from each Markov chain after discarding burn-in (to ensure stationarity). This provides a frequentist sampling distribution of $\hat \sigma_{\BM, f}^2(n, a_n, b_n)$ for a given iteration size $n$, batch size $b_n$ and number of batches $a_n$.  The whole experiment is then repeated for increasing values of $n$ to empirically assess the limiting behavior of the corresponding sampling distributions.  
    
    We consider two examples -- a simulated \emph{toy example} (Section~\ref{sec_illus_toy}) with a Markov chain for which the true (population) MCMC variance is known, and a \emph{real example} (Section~\ref{sec_illus_blasso}) with a practically useful Markov chain used that aids Bayesian inference in a  high-dimensional linear regression framework. The former illustrates the validity and accuracy of the CLT while the latter illustrates applicability of our results in real world scenarios. All computations in this section are done in \texttt{R v3.4.4} \citep{R_soft}, and the packages \texttt{tidyverse} \citep{pkg_tidyverse} and \texttt{flare} \citep{pkg_flare} are used. 
    
    \subsection{Toy example: Gibbs sampler with normal conditional distributions} \label{sec_illus_toy}
    
    In this section we consider a two-block toy normal Gibbs sampling Markov chain $(x_n, z_n)_{n \geq 0}$ with a state space $\R^2$ and transition $x \mid z \sim \N(z, 1/4)$ and $z \mid x \sim \N(x/2, 1/8)$. Our interest lies in the $x$-subchain, which evolves as $x_{n+1} = x_n/2 + N(0, 3/8)$. We consider the identity function $f(x) = x$, and seek to estimate the corresponding MCMC variance.  The example has been considered multiple times in the literature \citep{diaconis:khare:saloff:2008, qin:2019, chakraborty:khare:2019} and many operator theoretic properties of the chain have been thoroughly examined.  In particular, the eigenvalues of the associated Markov operator have been obtained as $(2^{-n})_{n \geq 0}$ \citep{diaconis:khare:saloff:2008}. This, together with reversibility of the Markov chain (since the marginal chain of a two-block Gibbs sampler is always reversible, \citep{geyer:1992}) implies geometric ergodicity. It is straight-forward to see that the target stationary distribution $\pi$ is the normal distribution $\N(0, 1/2)$, and the $h$-th order auto-covariance for the $x$ chain, $h \geq 0$, is given by $\gamma_h = \cov_\pi(x_{h}, x_0) = \langle f - E_\pi f, K^h (f - E_\pi f) \rangle =  2^{-(1+h)}$. Consequently, the true (population) MCMC variance of the chain is given by 
    \[
    \sigma^2_f = \gamma_0 + 2 \sum_{h = 1}^\infty \gamma_h = \frac{1}{2} + \sum_{h = 1}^\infty \frac{1}{2^h} = \frac{1}{2} + 1 = 1.5.
    \]
    
    To assess the asymptotic performances of the batch means estimator in this toy example, we generate 5,000 replicates of the proposed Markov chain, each with an iteration size of 520,000 and an independent standard normal starting point for $x$.  In each replicate, after throwing away the initial 20,000 iterations as burn-in, we compute the batch means estimate $\sigma_{\BM, f}^2 (n, a_n, b_n)$ for (i) $b_n = \sqrt{n}$, (ii) $b_n = n^{0.4}$ and (iii) $b_n = n^{1/3 + 10^{-5}}$ separately with the first (after burn-in)  $n = $ 5000, 10,000, 50,000, 100,000 and 500,000 iterations. The estimates are subsequently standardized by the population mean $\sigma^2_f = 1.5$ and the corresponding population standard deviations $\sqrt{2} \sigma_f^2/\sqrt{a_n} = 1.5 \sqrt{2/a_n}$. For each $n$, these standardized estimates from different replicates are then collected and their frequentist sampling distributions are  plotted as separate histograms for different choices of $b_n$ (blue histograms for $b_n = \sqrt{n}$, red histograms for $b_n = n^{0.4}$, and orange histograms for $b_n = n^{1/3 + 10^{-5}}$). These histograms, along with overlaid standard normal curves, are displayed  in Figure~\ref{fig:toynormalfigures}.

    \begin{figure}[ht]
        \centering
        \includegraphics[width=\linewidth]{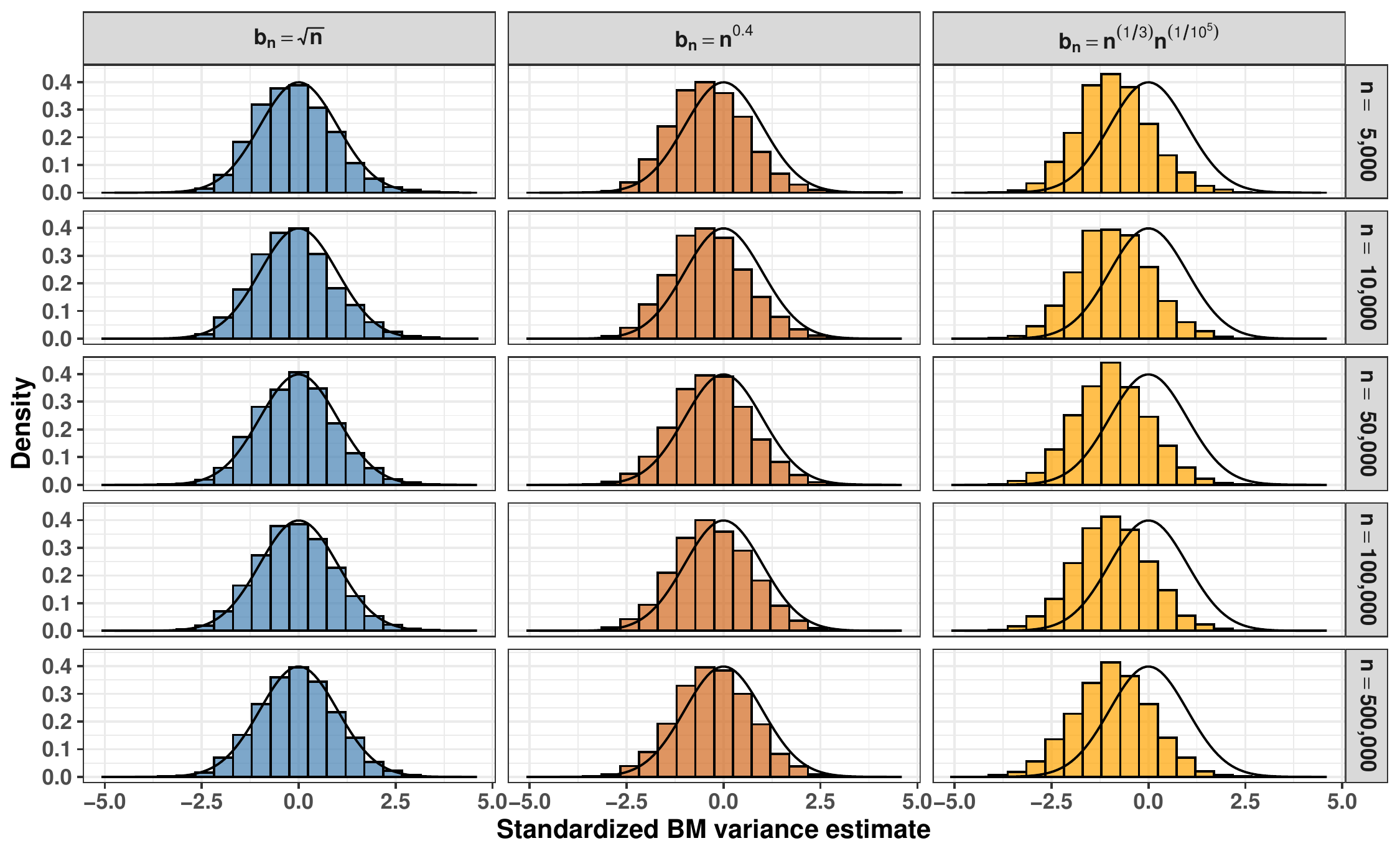}
        \caption{Frequentist sampling distribution of the batch means MCMC variance estimator in the toy normal example. The  sampling distribution of the standardized (with mean = $\sigma^2_f = 1.5$ and standard deviation = $\sqrt{2/a_n} \sigma^2_f = \sqrt{2/a_n} 1.5$) batch means MCMC variance estimator $\hat \sigma^2_{\BM, f}$ for the $x$-subchain obtained from 5,000 replicates are plotted as a matrix of histograms for various choices of $n$ and $b_n$. For each $n \in \{$5,000, 10,000, 50,000, 100,000, 500,000\} (plotted along the vertical direction of the histogram matrix), the blue histogram (left most panel) corresponds to $b_n = \sqrt{n}$, red (middle panel) corresponds to $b_n = n^{0.4}$ and orange (right most panel) corresponds to $b_n = n^{1/3 + 10^{-5}}$. The overlaid black curve on each histogram corresponds to the standard normal density function.}
        \label{fig:toynormalfigures}
    \end{figure}
    
    From Figure~\ref{fig:toynormalfigures}, the following observations are made. First, as $n \to \infty$ the sampling distributions of the BM variance estimates appear to become more ``normal'', i.e., the histograms become more symmetric and bell shaped, for all choices of $b_n$. This is a direct consequence of the CLT proved in Theorem~\ref{thm_clt_bme}. Second, of the three choices of $b_n$ considered, the BM variance estimates associated with $b_n = \sqrt{n}$  are the least biased, followed by $b_n = n^{0.4}$, and the estimates associated with $b_n = n^{1/3 + 10^{-5}}$ are the most biased. This is not surprising, as \eqref{bias_term_upper}  and \eqref{bias_term_lower} show that the asymptotic bias is of the same order of 
    $\sqrt{a_n}/b_n$.  As $n \to \infty$, the bias goes to zero, a fact that is well illustrated through the histograms for $b_n = \sqrt{n}$ (blue histograms) and $b_n = n^{0.4}$ (red histograms). For $b_n = n^{1/3 + 10^{-5}}$ (orange histograms) a much larger $n$ is required.
    
    Finally, to assess the practical utility of the proposed CLT, we note frequentist empirical coverage of approximate normal confidence intervals for the true MCMC variance $\sigma^2_{f}$ . In each replicate for each $(n, b_n)$ pair we first construct a 95\% approximate normal confidence interval with bounds $\hat \sigma_{\BM, f}^2(n, a_n, b_n) \pm 1.96 \sqrt{2/a_n} \hat \sigma_{\BM, f}^2(n, a_n, b_n)$. Then we compute the frequentist coverages of these 95\% confidence intervals by evaluating the proportion of replicates where the corresponding interval contains the true $\sigma^2_{f} = 1.5$, separately for each for each $(n, b_n)$ pair. These frequentist coverages are displayed in Table~\ref{tab:freq_coverage}, which shows near perfect coverage for $b_n = \sqrt{n}$ even for moderate $n$  ($\geq  50,000$), increasingly better coverage for $b_n = n^{0.4}$ (with moderately large $n$), and poor coverage for $b_n = n^{1/3 + 10^{-5}}$ even for large $n$ ($= 500,000$). These results are in concordance with the histograms displayed in Figure~\ref{fig:toynormalfigures}, and demonstrates that for the current problem $b_n = \sqrt{n}$ provides the fastest asymptotic normal convergence among the three choices of $b_n$ considered.  
    
    \begin{table}[ht]
        \centering
        \begin{tabular}{|r|c|c|c|}
            \toprule
            $n$ & $b_n = \sqrt{n}$ & $b_n = n^{0.4}$ & $b_n = n^{1/3 + 10^{-5}}$ \\ 
            \midrule
            5,000 & 0.924 & 0.902 & 0.814 \\ 
            10,000 & 0.927 & 0.907 & 0.810 \\ 
            50,000 & 0.946 & 0.932 & 0.825 \\ 
            100,000 & 0.943 & 0.934 & 0.835 \\ 
            500,000 & 0.949 & 0.941 & 0.834 \\ 
            \bottomrule
        \end{tabular}
        
        \caption{Frequentist coverages of approximate normal 95\% confidence intervals for the MCMC variance $\sigma^2_f$ based on the batch means estimator $\sigma_{\BM, f}^2(n, a_n, b_n)$ for various choices of $n$ and $b_n$.}
        \label{tab:freq_coverage}
        
    \end{table}

    \subsection{Real data example: data augmentation Gibbs sampler for Bayesian lasso regression} \label{sec_illus_blasso}

    This section illustrates the applicability of the proposed CLT in a real world application. Consider the linear regression model 
    \[
    Y \mid \mu, \beta, \eta \sim \N_m (\mu + X\beta,  \eta^2 I_m)
    \]
    where $Y \in \R^n$ is a vector of responses, $X$ is a non-stochastic $m \times p$ design matrix of standardized covariates, $\beta \in \R^p$ is a vector of unknown regression coefficients, $\eta^2 > 0$ is an unknown residual variance,  $\mu \in \R$ is an unknown intercept, $\N_d$ denotes the $d$-variate ($d \geq 1$) normal distribution
    and $I_m$ denotes the $m$-dimensional identity matrix. Interest lies in the estimation of $\beta$ and $\eta^2$. In many modern-day applications, the sample size $m$ is smaller than the number $p$ of covariates. For a meaningful estimation of $\beta$ in such a scenario regularization (i.e., shrinkage towards zero) of the estimate is necessary. A particularly useful regularization approach involves the use of a lasso penalty \citep{tibshirani:1996}, producing lasso estimates of the regression coefficients. The Bayesian lasso framework \citep{park:casella:2008} provides a probabilistic approach to quantifying uncertainties in the lasso estimation. Here, one considers the following hierarchical priors for $\beta$:
    \begin{align*}
    \beta &\sim \N_p(0, \eta^2  D_\tau) \\
    \tau_j &\sim \text{i.i.d. Exponential(rate = }\lambda^2/2)
    \end{align*}
    and estimates $\beta$ through the associated posterior distribution obtained from the Bayes rule:
    \[
    \text{posterior density} \propto \text{prior density} \times \text{likelihood}.
    \]
    Here  $D_\tau$ is the diagonal matrix $\diag\{\tau_1, \dots, \tau_p\}$, and $\lambda > 0$ is a prior hyper-parameter that determines the amount of sparsity in $\beta$. Note that the marginal (obtained by integrating out $\tau_j$'s) prior for $\beta$ is a product of independent Laplace densities, and the associated marginal posterior mode of $\beta$ corresponds to the frequentist lasso estimate of $\beta$. 
    
    It is clear that the target posterior distribution of $\beta$, $\sigma$ and $\tau = (\tau_1, \dots, \tau_p)$ is intractable, i.e., it is not avaialable in closed form, and i.i.d. random generation from the distribution is infeasible. \citet{park:casella:2008} suggested a three-block Gibbs sampler for MCMC sampling from the target posterior which was later shown to be geometrically ergodic \citep{khare:hobert:2013:blasso}. A more efficient (in an operator theoretic sense) two-block version of this three-block Gibbs sampler has been recently proposed in \citet{rajarantnam:sparks:khare:zhang:2017}, where the authors prove the \textit{trace-class} property of the proposed algorithm, which in particular, also implies geometric ergodicity  (recall that a two-block Gibbs sampler is always reversible). One iteration of the proposed two-block Gibbs sampler consists of the following random generations.
    \begin{enumerate}
        \item Generate $(\beta, \eta^2)$ from the following conditional distributions: 
        \begin{align*}
        \eta^2 \mid \tau, Y  & \sim \text{Inverse-Gamma} \displaystyle \left(\frac{(m + p - 1)}{2}, \frac12 \left\| \tilde Y - X \beta \right\|^2 + \frac12 \beta^T D_\tau^{-1} \beta/2 \right) \\
        \beta \mid \eta^2, \tau, Y & \sim \N_p \left(A_\tau^{-1} X^T \tilde{Y}, \eta^2 A_\tau^{-1} \right).
        \end{align*}
        
        \item Independently generate  $\tau_1, \dots, \tau_p$ such that the full conditional distribution of $1/\tau_j$, $j = 1, \dots, p$ is given by
        \[
        1/\tau_j \mid \beta, \eta^2, Y \sim \text{Inverse-Gaussian} \left( \sqrt{\frac{\lambda \eta^2}{\beta_j^2}}, \lambda \right).
        \]
    \end{enumerate}
    Here $\tilde Y = Y - m^{-1} (Y^T 1_m) 1_m$, $1_m$ being the $m$-component vector of 1's, and $A_\tau = X^TX + D_\tau^{-1}$.

    For a real world application of the above sampler we consider the gene expression data of \citet{scheetz:2006}, made publicly available in the \texttt{R} package \texttt{flare} \citep{flare_pkg} as the data set entitled \texttt{eyedata}. The data set consists of $m = 120$ observations on a response variable (expression level) and $p = 200$ predictor variables (gene probes). \citet{rajarantnam:sparks:khare:zhang:2017} analyze this data set in the context of the Bayesian lasso regression, and provide an efficient \texttt{R} implementation of the aforementioned two-block Gibbs sampler in their supplementary document. Following \citep{rajarantnam:sparks:khare:zhang:2017} we standardize the columns of design matrix $X$ and choose the prior (sparsity) hyperparameter as $\lambda = 0.2185$ which ensures that the frequentist lasso estimate (marginal posterior mode) of $\beta$ has $\min\{m, p\}/2 = 60$ non-zero elements. 
    
    We focus on the marginal $(\beta, \eta^2)$ chain of the Bayesian lasso Gibbs sampler described above. This marginal chain is reversible, and we seek to estimate the MCMC variance of the linear regression log-likelihood function 
    \[
    f(\beta, \eta^2, \tau) = -\frac{m}{2} \log (\eta^2) - \frac{1}{2 \eta^2} \| \tilde Y - X\beta \|_2^2 
    \]
    using the batch means variance estimator. To empirically assess the asymptotic behavior of this estimator, we obtain its frequentist sampling distribution as described in the following.  We generate 5,000 replicates of the above Markov chain with independent random starting points (the initial $\beta$ is generated from a standard multivariate normal  distribution and the initial $\eta^2$ is generated from an independent standard exponential distribution). The \texttt{R} script provided in the supplementary document in \cite{rajarantnam:sparks:khare:zhang:2017} is used for the Markov chain generations. On each replicate we run 120,000 iterations of the Markov chain, discard the initial 20,000 iterations as burn-in, and evaluate the log-likelihood at the remaining 100,000 iterations. The BM variance estimator $\sigma_{\BM, f}^2$ is subsequently computed from the evaluated log-likelihood $f$ at the first $n =$  5,000, 10,000, 50,000 and 100,000  iterations and for $b_n = \sqrt{n}$, $b_n = n^{0.4}$ and $b_n = n^{1/3 + 10^{-5}}$, and the resulting replicated estimates are then collected for each $(n, b_n)$ pair. Since the true MCMC variance $\sigma^2_{\BM, f}$ is of course unknown here, we focus on the asymptotic normality of only approximately standardized estimates over replications. More specifically, we first evaluate the mean (over replications) batch means estimate
    \[
    \bar {\hat{\sigma}^2}_{\BM, f}(n = 100,000, a_n, b_n) = \frac{1}{5000} \sum_{l = 1}^{5000} \hat \sigma^2_{\BM, f}(n = 100,000, a_n, b_n)_{(l)}
    \]
    where for each $b_n$ (and hence $a_n$) $\bar {\hat{\sigma}^2}_{\BM, f}(n = 100,000, a_n, b_n)_{(l)}$ denotes the corresponding batch means variance estimate obtained from the $l$th replicate with $n = 100,000$, $l = 1, \dots, 5000$. The estimates $\bar {\hat{\sigma}^2}_{\BM, f}(n = 100,000, a_n, b_n)$ for the above three choices of $b_n$ are displayed in Table~\ref{tab:blasso_sigma2_est}. 
    
    \begin{table}[ht]
    	\centering
    	\begin{tabular}{|r|r|r|r|}
    		\toprule
    		$b_n$ & $\sqrt{n}$ & $n^{0.4}$ &  $n^{1/3 + 10^{-5}}$ \\
    		\midrule 
    		 $\bar {\hat{\sigma}^2}_{\BM, f}(n = 100,000, a_n, b_n)$ & 304.351 & 302.385 & 299.091 \\ 
    		\bottomrule
    	\end{tabular}
    \caption{The mean (over 5000 replications) batch means estimate $\bar {\hat{\sigma}^2}_{\BM, f}(n = 100,000, a_n, b_n)$ of $\sigma_f^2$ obtained from replicated MCMC draws each with iteration size $n =$ 100,000 and batch sizes $b_n = \sqrt{n}$, $n^{0.4}$ and $n^{1/3 + 10^{-5}}$. }
    \label{tab:blasso_sigma2_est}
    \end{table}
    
    After computing $\bar {\hat{\sigma}^2}_{\BM, f}(n = 100,000, a_n, b_n)_{(l)}$, we standardize all replicated batch means estimates with mean = $\bar {\hat{\sigma}^2}_{\BM, f}(n = 100,000, a_n, b_n)$ and standard deviation = $\bar {\hat{\sigma}^2}_{\BM, f}(n = 100,000, a_n, b_n) \sqrt{2/a_n}$ separately  for each  $(n, b_n)$ pair. The frequentist sampling distributions of these \textit{approximately} standardized estimates are plotted as a matrix of histograms for various choices of $n$ and $b_n$, along with overlaid standard normal density curves, in Figure~\ref{fig:blassofigures}. From the figure, it follows that these sampling distributions of the approximately standardized estimates are very closely approximated by a standard normal distribution.  Of course, unlike the histograms displayed in Figure~\ref{fig:toynormalfigures} for the toy normal example (Section~\ref{sec_illus_toy}), no information on the bias of the estimates can be obtained here. However, these histograms  do demonstrate the remarkable accuracy of an asymptotic normal approximation, and thus illustrates the applicability of the proposed CLT for the batch means MCMC variance estimate in a real world application.

    \begin{figure}[ht]
        \centering
        \includegraphics[width=\linewidth]{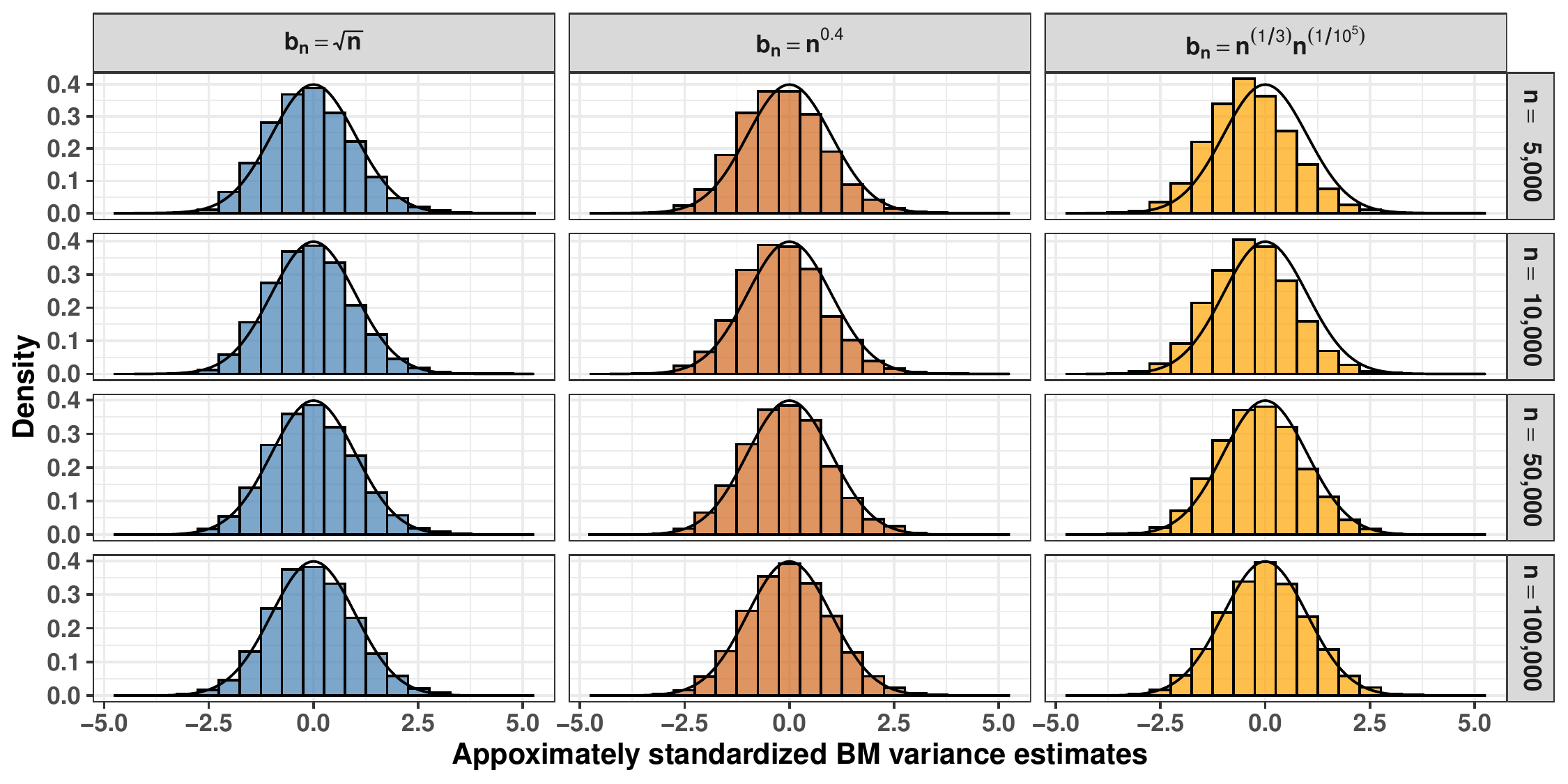}
        \caption{Frequentist sampling distribution of the batch means MCMC variance estimator in the Bayesian lasso example. The  sampling distribution of the \emph{approximately} standardized (with mean =  $\bar {\hat{\sigma}^2}_{\BM, f}(n = 100,000, a_n, b_n)$ and standard deviation = $\bar {\hat{\sigma}^2}_{\BM, f}(n = 100,000, a_n, b_n) \sqrt{2/a_n}$, see Table~\ref{tab:blasso_sigma2_est}) batch means MCMC variance estimator $\hat \sigma^2_{\BM, f}(n, a_n, b_n)$ for the linear regression log-likelihood function $f$ evaluated at the iterations of the Bayesian lasso two block Gibbs sampler are plotted as matrix of histograms for various choices of $n$ and $b_n$. For each $n \in \{$5,000, 10,000, 50,000, 100,000, 500,000\} (plotted along the vertical direction of the histogram matrix), the blue histogram (left most panel) corresponds to $b_n = \sqrt{n}$, red (middle panel) corresponds to $b_n = n^{0.4}$ and orange (right most panel) corresponds to $b_n = n^{1/3 + 10^{-5}}$. The overlaid black curve on each histogram corresponds to the standard normal density function.}
        \label{fig:blassofigures}
    \end{figure}

    \bibliographystyle{apa}
    \bibliography{batchmeans_clt_references}

    \begin{appendix}
        
        \section{Proofs of Results used in Lemma~{\ref{LEMMA_2}}} \label{sec_appen_proof}
        
        \begin{proposition} \label{prop_Exi2givenF}
            Consider $\xi_{k, n}$ as defined in \eqref{xi_defn}, and assume that the assumptions in Theorem~\ref{thm_clt_bme} hold. Then 
            \[
            \frac{1}{a_n - 1} \sum_{k=2}^{a_n} E (\xi_{k, n}^2 \mid \mathcal{F}_{k-1, n}) \xrightarrow{P} 1.
            \]
        \end{proposition}
        
        \begin{proof}
            Observe that, due to the Markov property of $\mcx$, $E (\xi_{k, n}^2 \mid \mathcal{F}_{k-1, n})$ is a function only of $X_{(k-1)b_n}$, for all $k = 2, \dots, a_n$. Define $\check h(X_{(k-1)b_n}) = E (\xi_{k, n}^2 \mid \mathcal{F}_{k-1, n}) - 1$, with $\check h(X_{kb_n})  \in L_0^2(\pi)$ for all $k, n$, as $E_\pi (f^8) < \infty$ and $E_\pi (\xi_{k, n}^2) = 1$. It is enough to show that the mean squared convergence
            \[
            E_\pi \left[\frac{1}{a_n - 1} \sum_{k = 1}^{a_n - 1} \check h(X_{(k-1)b_n}) \right]^2 \to 0
            \]
            holds. To this end, note that
            \begin{align*} \label{E_mean2_check_h}
            &\quad E_\pi \left[\frac{1}{a_n - 1} \check h(X_{(k-1)b_n}) \right]^2  \\
            &= \frac{1}{(a_n - 1)^2} \sum_{k=1}^{a_n - 1} E_\pi \left[\check h(X_{kb_n})^2 \right] + 
            \frac{2}{(a_n - 1)^2} \mathop{\sum \sum}_{1 \leq k < k' \leq a_n - 1} E_\pi \left[\check h(X_{kb_n}) \check h(X_{k'b_n}) \right]. \numbereqn
            \end{align*} 
            Due to stationarity of $\mcx$, $E_\pi \left[\check h(X_{kb_n})^2 \right]$ is the same for all $k \geq 1$, say $B = \|\check h\|_\pi^2 =  E_\pi \left[\check h(X_{kb_n})^2 \right]$, where $B < \infty$ as $E_\pi (f^8) < \infty$. Consequently  
            \[
            \frac{1}{(a_n - 1)^2} \sum_{k=1}^{a_n - 1} E_\pi \left[\check h(X_{kb_n})^2 \right] = \frac{1}{a_n} \|\check h \|^2_\pi \to 0
            \]
            as $n \to \infty$,  and it remains to show that the second term in \eqref{E_mean2_check_h} also converges to zero. Note that,
            \begin{align*}
            &\quad \left| \frac{1}{(a_n - 1)^2}\mathop{\sum \sum}_{1 \leq k < k' \leq a_n - 1} E_\pi \left[\check h(X_{kb_n}) \check h(X_{k'b_n}) \right] \right| \\
            &\leq   \frac{1}{(a_n - 1)^2}\mathop{\sum \sum}_{1 \leq k < k' \leq a_n - 1} \left| E_\pi \left[\check h(X_{kb_n}) E \left(\check h(X_{k'b_n}\right) \mid X_{kb_n}) \right] \right|  \\
            &= \frac{1}{(a_n - 1)^2}\mathop{\sum \sum}_{1 \leq k < k' \leq a_n - 1} \left| E_\pi \left[\check h(X_{kb_n})  \left(K^{b_n (k'-k)} h \right) (X_{kb_n}) \right] \right| \\
            &= \frac{1}{(a_n - 1)^2}\mathop{\sum \sum}_{1 \leq k < k' \leq a_n - 1}  \left| \left\langle \check h,  K^{b_n (k'-k)} \check h \right\rangle_\pi \right| \\
            & \stackrel{(\star)}{\leq} \frac{1}{(a_n - 1)^2}\mathop{\sum \sum}_{1 \leq k < k' \leq a_n - 1}  \left\| \check h \right\|_\pi  \left\| K^{b_n (k'-k)} \check h \right\|_\pi \\
            & \stackrel{(\star \star)}{\leq} \frac{1}{(a_n - 1)^2} \mathop{\sum \sum}_{1 \leq k < k' \leq a_n - 1} \|\check h\|_\pi^2 \lambda^{b_n (k' - k)} \\
            & = \frac{1}{(a_n - 1)^2}\left\| \check h \right\|_\pi^2 \sum_{k = 1}^{a_n - 1}  \sum_{r = 1}^{a_n-1-k} \lambda^{rb_n} \\
            & \leq \frac{\left\| \check h \right\|_\pi^2}{(a_n - 1)^2}  \sum_{k = 1}^{a_n - 1}  \sum_{r = 1}^{\infty} \lambda^{rb_n} \\
            & =  \frac{\left\| \check h \right\|_\pi^2}{(a_n - 1)} \frac{1}{1 - \lambda^{rb_n}}
            \to 0 
            \end{align*}
            as $n \to \infty$, where $(\star)$ follows from the Schwarz inequality, and $(\star \star)$ follows from the  operator norm inequality $\|K \check h\|_\pi \leq \|K\| \|h\|_\pi$, and as before we let $\lambda = \|K\|$ with $\lambda \in (0, 1)$ due to geometric ergodicity of $\mcx$. This completes the proof.
            
        \end{proof}

        \begin{proposition} \label{prop_Y1bar4}
            Under the setup assumed in Theorem~\ref{thm_clt_bme}, we have $E_\pi \left(b_n^2 \bar Y_k^4\right) \to 3 \sigma_f^4$ as $n \to \infty$,  for each $k = 1, \dots, a_n$.
        \end{proposition}
        
        \begin{proof}
            On the outset, note that since $(X_n)_{n \geq 0}$ is stationary, $E_\pi(\bar Y_1^4) = E_\pi(\bar Y_k^4)$. Moreover, since
            $b_n \to \infty$ as $n \to \infty$, it is therefore enough to show that as $n \to \infty$, 
            \[
            \frac{1}{n^2} E_\pi\left(Y_1 + Y_2 + \dots + Y_n \right)^4\to 3 \sigma_f^4.
            \]
            For the remainder of the proof, we shall therefore replace $b_n$ by $n$. We will proceed by expanding $E_\pi(Y_1 + Y_2 + \cdots + Y_n)^4$ and analyzing relevant terms separately. First, let us define $\mu_r = E_\pi(Y_1^r) = E_\pi [f(X_i) - E_\pi f]^r$ for $r = 2, 4, 6$. Note that $E_\pi (f^8) < \infty$ implies that $\mu_r < \infty$ for all $r = 2, 4, 6$. Now observe that, 
            \begin{align*}
            & \quad \frac{1}{n^2} E_\pi \left(Y_1+Y_2+\cdots+Y_n\right)^4\\
            &= \frac{1}{n^2} E_\pi \left(\sum_{i=1}^n Y_i^4 + 4\sum_{i \neq j} Y_i^3 Y_j + 6\sum_{i<j} Y_i^2 Y_j^2 + 12\sum_{i \neq j \neq k, j < k} Y_i^2 Y_j Y_k + \sum_{i \neq j \neq k \neq l} Y_i Y_j Y_k Y_l \right)\\
            &= U_1 + U_2 + U_3 + U_4 + U_5, \text{ say,}
            \end{align*}
            and we shall consider the convergence of each $U_i$, $i = 1, \dots, 5$  separately. Since $\mu_4 < \infty$ and the chain is stationary, it follows that $E_\pi (n^{-1} \sum_{i=1}^{n} Y_i^4) = \mu_4$ for all $n$, so that
            \begin{equation} \label{U_1_limit}
            U_1 =  E_\pi \left(\frac{1}{n^{2}} \sum_{i=1}^{n} Y_i^4 \right) = \frac{1}{n} E_\pi \left(\frac{1}{n} \sum_{i=1}^{n} Y_i^4 \right) \to 0.
            \end{equation}
            
            \noindent As for $U_2$, note that,
            \begin{align*}
            |U_2| 
            &= \frac{4}{n^2} \left|\sum_{i\neq j} E_\pi \left(Y_i^3 Y_j\right) \right| \\
            &\leq \frac{4}{n^2}\sum_{i\neq j} \left| E_\pi \left(Y_i^3 Y_j\right) \right| \\
            &= \frac{8}{n^2} \sum_{i<j} \left| E_\pi \left(Y_i^3 Y_j\right) \right| \\ 
            &= \frac{8}{n^2} \sum_{i<j} \left| E_\pi E \left[ f_0^3(X_i) f_0(X_j) \mid X_i \right] \right| \\
            &= \frac{8}{n^2} \sum_{i<j} \left|E_\pi \left[f_0^3(X_i) K^{j-i} f_0(X_i) \right] \right| \\
            &\leq \frac{8}{n^2} \sum_{i<j} E_\pi\left|f_0^3(X_i) K^{j-i}f_0(X_j)\right| \\
            &\stackrel{(\star_1)}{\leq} \frac{8}{n^2} \sum_{i<j} \left[E_\pi \left|f_0^3(X_i)\right|^{\frac{4}{3}} \right]^{\frac{3}{4}} \left[E_\pi\left|K^{j-i}f(X_i)\right|^4\right]^{\frac{1}{4}} \\
            &\leq \frac{8}{n^2}\sum_{i<j} \left[E_\pi f_0^4(X_i)\right]^{\frac{3}{4}} \lambda^{j-i} \left[E_\pi f_0^4(X_i)\right]^{\frac{1}{4}} \\
            &= \frac{8}{n^2} \ \mu_4 \sum_{i<j} \lambda^{j-i} \\
            &= \frac{8}{n^2} \ \mu_4 \sum_{k=1}^{n-1} (n-k) \lambda^k \\
            &\leq \frac{8}{n^2} \sum_{k=1}^\infty n \lambda^k
            = \frac{8}{n} \mu_4\ \lambda \left(1-\lambda\right)^{-1} 
            \end{align*}
            Here, as defined in Lemma~\ref{LEMMA_1}, $f_0 = f - E_\pi f \in L^2_0(\pi)$, $\lambda = \|K\| \in (0, 1)$, and $(\star_1)$ is a consequence of H\"older's inequality. Thus,
            \begin{equation} \label{U_2_limit}
            U_2 \to 0 \text{ as } n \to \infty.
            \end{equation}
            
            \noindent Next we focus on $U_3$. Since
            \begin{equation*}
            E_\pi\left[Y_i^2 Y_j^2\right] 
            = E_\pi\left[Y_i^2 \left(Y_j^2-\mu_2\right) \right] + \mu_2^2 
            = E_\pi\left[f_0^2(X_i) \tilde{S}(X_j) \right] + \mu_2^2
            \end{equation*}
            where $\tilde{S}(x)=f_0^2(x) - \mu_2\in L_0^2(\pi)$, therefore,  
            \begin{align*}
            U_3 = \frac{6}{n^2} \sum_{i<j} E_\pi \left(Y_i^2 Y_j^2\right) 
            &= \frac{6}{n^2}\sum_{i<j}E_\pi\left[f_0^2(X_i) \tilde{S}(X_j)\right] + \frac{6}{n^2} \frac{n(n-1)}{2}\mu_2^2 \\
            &= \frac{6}{n^2} \sum_{i<j}E_\pi\left[f_0^2(X_i) \tilde{S}(X_j)\right] + 3\left(1 - \frac{1}{n}\right) \mu_2^2.
            \end{align*}
            
            \noindent Now
            \begin{align*}
            \left| \frac{6}{n^2} \sum_{i<j} E_\pi \left[f_0^2(X_i) \tilde{S}(X_j) \right] \right|
            &\leq \frac{6}{n^2} \sum_{i<j} \left| E_\pi \left[f_0^2(X_i) \tilde{S}(X_j)\right]\right| \\
            &=\frac{6}{n^2} \sum_{i<j} \left|E_\pi\left[ f_0^2(X_i) E\left[ \tilde{S}(X_j) \mid X_i\right] \right] \right|\\
            &= \frac{6}{n^2} \sum_{i<j} \left|E_\pi\left[f_0^2(X_i) K^{j-i} \tilde{S}(X_i) \right] \right| \\
            &\stackrel{(\star_2)}{\leq} \frac{6}{n^2}  \sum_{i<j} \left[E_\pi f_0^4(X_i) \right]^{\frac{1}{2}} \left[E_\pi \left[K^{j-i} \tilde{S}(X_i) \right]^2 \right]^{\frac{1}{2}} \\
            &\leq \frac{6}{n^2}  \sum_{i<j} \left[E_\pi f_0^4(X_i) \right]^{\frac{1}{2}} \|K\|^{j-i} \left[E_\pi \left(f_0^2(X_i) - \mu_2\right)^2 \right]^{\frac{1}{2}} \\
            &\leq \frac{6}{n^2}  \sum_{i<j} \left[E_\pi f_0^4(X_i) \right]^{\frac{1}{2}} \|K\|^{j-i} \left[E_\pi f_0^4(X_i) \right]^{\frac{1}{2}} \\
            &\leq \frac{6}{n^2} \mu_4 \sum_{i<j} \lambda^{j-i} \\
            &\leq \frac{6}{n^2} \mu_4 \sum_{k=1}^n n\lambda^k \\
            &\leq \frac{6}{n} \mu_4 \ \lambda \left(1-\lambda\right)^{-1} \to 0\ \text{as }n\to\infty.
            \end{align*}
            Here $(\star_2)$ follows from Schwarz's inequality.  Consequently,
            \begin{equation}\label{U_3_limit}
            U_3 \to 3 \mu_2^2 \text{ as } n \to \infty.
            \end{equation}
            
            \noindent Next we  consider $U_4$. Observe that
            \begin{align*}
            U_4 
            &= \frac{12}{n^2} \left[\sum_{i<j<k} E_\pi \left(Y_i^2 Y_j Y_k \right) + \sum_{j<k<i} E_\pi \left(Y_i^2 Y_j Y_k \right) + \sum_{j<i<k} E_\pi \left(Y_i^2 Y_j Y_k \right)\right] \\
            &= \frac{12}{n^2} \left[\sum_{i<j<k} \mu_2\ E_\pi\left(Y_jY_k\right) + \sum_{j<k<i} \mu_2\ E_\pi\left(Y_j Y_k\right) \right] \\
            & \qquad + \frac{12}{n^2}\left[ \sum_{i<j<k} E_\pi\left[\widetilde{Y_i^2} Y_j Y_k\right] + \sum_{j<k<i} E_\pi\left[\widetilde{Y_i^2} Y_j Y_k\right] + \sum_{j<i<k} E_\pi\left[Y_i^2 Y_j Y_k \right]\right] \\
            &= U_4^{(1)} + U_4^{(2)}, \text{ say}.
            \end{align*}
            Here $\widetilde{Y_i^2} = \tilde S(X_i)$ =  $Y_i^2 - \mu_2 \in L^2(\pi)$. Note that
            \begin{align*}
            \ U_4^{(1)} &= \frac{12}{n^2}\ \mu_2 \left[\sum_{i<j<k} \langle f_0, K^{j-k} f_0 \rangle_\pi + \sum_{j<k<i} \langle f_0, K^{j-k} f_0 \rangle_\pi \right] \\
            &= \frac{12}{n^2}\ \mu_2 \sum_{r=1}^{n-2}(n-r-1)(n-r)\ \langle f_0, K^r f_0 \rangle_\pi \\
            &= 12\ \mu_2 \sum_{r=1}^{n-2} \left(1 - \frac{r-1}{n}\right) \left(1 - \frac{r}{n}\right) \langle f_0, K^r f_0 \rangle_\pi \\
            &\to 12\ \mu_2 \sum_{r=1}^{\infty} \langle f_0, K^r f_0 \rangle_\pi
            = 12\ \mu_2 \sum_{r = 1}^\infty \gamma_r \numbereqn \label{U_4_1_limit}
            \end{align*}
            as $n \to \infty$, where $\gamma_h$'s are the auto-covariances as defined in \eqref{defn_gamma_h}, and  the last convergence follows from the dominated convergence theorem. As for $U_4^{(2)}$, observe that
            \begin{equation} \label{U_4_2_upperbd}
            \left| U_4^{(2)} \right| \leq \frac{12}{n^2} \left[\sum_{i<j<k} \left|E_\pi \left(\widetilde{Y_i^2} Y_j Y_k\right) \right| + \sum_{j<k<i} \left|E_\pi \left(\widetilde{Y_i^2} Y_j Y_k\right) \right| + \sum_{j<i<k} \left| E_\pi \left(Y_i^2 Y_j Y_k \right) \right| \right].
            \end{equation}
            For $i<j<k$,
            \begin{align*}
            \left|E_\pi\left(\widetilde{Y_i^2} Y_j Y_k\right)\right| 
            &\stackrel{(\star_3)}{=} \left|E_\pi\left[Y_j Y_k E\left(\widetilde{Y_i^2} 
            \mid X_j, X_k\right) \right] \right| \\
            & \stackrel{(\star_4)}{=} \left|E_\pi \left[Y_j Y_k E \left(\widetilde{Y_i^2} \mid X_j \right) \right] \right| \\
            &\leq E_\pi \left|Y_j Y_k K^{j-i} \tilde{Y_j}^2 \right| \\
            &\stackrel{(\star_5)}{\leq} \left[E_\pi \left(Y_j^2 Y_k^2 \right)\right]^{\frac{1}{2}} \left[E_\pi\left( K^{j-i}\tilde{Y_j}^2\right)^2\right]^{\frac{1}{2}} \\ 
            &\stackrel{(\star_6)}{\leq} \sqrt{\left[E_\pi \left(Y_j^4 \right)\right]^{\frac{1}{2}} \left[ E_\pi \left(Y_k^4 \right) \right]^{\frac{1}{2}}}\  \lambda^{j-i} \sqrt{E_\pi \left( \tilde{Y_j}^4\right)} \\
            &\leq 4 \lambda^{j-i} \mu_4 \numbereqn \label{E_tYi2YjYk_upper_bd_1}.
            \end{align*}
            Here $(\star_3)$ and $(\star_4)$ are consequences of reversibility and Markov property respectively, and $(\star_5)$ and $(\star_6)$ are due to Schwarz's inequality. Again for $i < j < k$,
            \begin{align*}
            \left|E_\pi \left(\widetilde{Y_i^2} Y_j Y_k \right) \right| 
            &=\left|E_\pi \left[\widetilde{Y_i^2} Y_j E\left(Y_k \mid X_i,X_j \right)\right] \right| \\
            &\stackrel{(\star_7)}{=} \left|E_\pi \left[ \widetilde{Y_i^2} Y_j E\left(Y_k \mid X_j \right)\right]\right| \\
            &\leq E_\pi \left|\widetilde{Y_i^2} Y_j K^{k-j} Y_j\right| \\
            &\stackrel{(\star_8)}{=} 8 \lambda^{k-j} \sqrt{\mu_2 \mu_6} \numbereqn \label{E_tYi2YjYk_upper_bd_2} 
            \end{align*}
            where $(\star_7)$ is due to the Markov property, and $(\star_8)$ follows from H\"older's inequality. Therefore, from \eqref{E_tYi2YjYk_upper_bd_1} and \eqref{E_tYi2YjYk_upper_bd_2}, we get			
            \begin{align*}
            \left|E_\pi\left[\widetilde{Y_i^2}Y_jY_k\right]\right| 
            &\leq \min \left\{\lambda^{j-i}, \lambda^{k-j} \right\} \left(4 \mu_4 + 8 \sqrt{\mu_2 \mu_6}\right) \\
            &= \left(\sqrt{\lambda} \right)^{2\max \left\{j-i, k-j\right\}}\left(4\mu_4+8\sqrt{\mu_2\mu_6}\right) \\
            &\leq \left(\sqrt{\lambda}\right)^{k-i} \left(4\mu_4+8\sqrt{\mu_2\mu_6}\right) 
            \end{align*}
            where the last inequality is a consequence of the fact that for two real numbers $a$ and $b$, $a+b \leq 2 \max\{a, b\}$ and that $\lambda = \|K\| \in (0, 1)$. Hence,
            \begin{align*}
            \frac{12}{n^2} \sum_{i<j<k} \left|E_\pi \left( \widetilde{Y_i^2} Y_j Y_k \right) \right|
            &\leq \frac{12}{n^2} \left(4 \mu_4 + 8 \sqrt{\mu_2 \mu_6}\right) \sum_{i<j<k}\left(\sqrt{\lambda}\right)^{k-i} \\
            &\leq \frac{12}{n^2} \left(4 \mu_4 + 8 \sqrt{\mu_2 \mu_6}\right) \sum_{r=2}^{n-1} (n-r)(r-1) \left(\sqrt{\lambda} \right)^r \\
            &\leq \frac{12}{n} \left(4\mu_4 + 8 \sqrt{\mu_2 \mu_6}\right) \sum_{r=1}^\infty r \left(\sqrt{\lambda}\right)^r \\
            &= \frac{12}{n} \left(4 \mu_4 + 8 \sqrt{\mu_2 \mu_6}\right) \sqrt{\lambda} \left(1 - \sqrt{\lambda}\right)^{-2}\to 0\ \text{ as } n \to \infty. \\
            \end{align*}
            By similar arguments, it can be shown that
            \[
            \frac{12}{n^2} \sum_{j < k < i} \left|E_\pi \left( \widetilde{Y_i^2} Y_j Y_k \right) \right| \to 0,  \text{ and } \
            \frac{12}{n^2} \sum_{j < i < k} \left|E_\pi \left( \widetilde{Y_i^2} Y_j Y_k \right) \right| \to 0
            \]
            as $n \to \infty$, which, from \eqref{U_4_2_upperbd} implies,
            \begin{equation} \label{U_4_2_limit}
            U_4^{(2)} \to 0 \text{ as } n \to \infty. 
            \end{equation}
            
            \noindent It follows from \eqref{U_4_1_limit} and \eqref{U_4_2_limit} that 
            \begin{equation} \label{U_4_limit}
            U_4 = U_4^{(1)} + U_4^{(2)} \to 12\ \mu_2 \sum_{r = 1}^\infty \gamma_h \text{ as } n \to \infty.
            \end{equation}
            
            \noindent Finally, we focus on $U_5$. Note that
            \begin{align*}
            \ U_5 &= \frac{24}{n^2}\sum_{i<j<k<l} E_\pi\left(Y_i Y_j Y_k Y_l\right) \\
            &=  \frac{24}{n^2} \sum_{i<j<k<l} E_\pi\left(Y_i Y_j\right) E_\pi\left(Y_k Y_l\right) + \frac{24}{n^2} \sum_{i<j<k<l} E_\pi \left(\left[Y_i Y_j - E_\pi\left(Y_i Y_j\right)\right] Y_k Y_l\right) \\
            &= U_5^{(1)} + U_5^{(2)}, \text{ say}.
            \end{align*}
            Then,
            \begin{align*}
            \ U_5^{(1)}
            &=\frac{24}{n^2} \sum_{i<j<k<l} \langle f, K^{j-i} f\rangle_\pi \langle f, K^{l-k}f \rangle_\pi \\
            &= \frac{24}{n^2} \sum_{r=1}^{\floor*{\frac{n}{2}-1}} \frac{(n-2r-2)(n-2r-1)}{2} \langle f, K^rf \rangle_\pi^2 \\
            &\quad + \frac{24}{n^2}\sum_{2 \leq r+r' \leq n-2} \frac{\left[n-(r+r')-2\right] \left[n-(r+r')-1\right]}{2} \langle  f, K^rf \rangle_\pi \langle f, K^{r'} f\rangle_\pi \\
            &\xrightarrow{(\star_9)} 12\left[\sum_{r=1}^\infty \langle f, K^r f \rangle_\pi^2 + \sum_{r \neq r'} \langle f, K^r f\rangle_\pi \langle f, K^{r'} f \rangle_\pi \right] \\ 
            &= 3 \left(2\sum_{r=1}^\infty \langle f, K^r f \rangle_\pi \right)^2 = 3 \left(2\sum_{r=1}^\infty \gamma_r \right)^2 \numbereqn \label{U_5_1_limit}
            \end{align*}
            where $(\star_9)$ follows from the dominated convergence theorem. As for $U_5^{(2)}$, observe that
            \[
            \left|U_5^{(2)} \right| \le \frac{24}{n^2} \sum_{i<j<k<l} \left| E_\pi \left(\left[Y_i Y_j - E_\pi\left(Y_i Y_j\right)\right] Y_k Y_l\right) \right|.
            \]
            
            \noindent Now for $i<j<k<l$, 
            \begin{align*}
            \left| E_\pi \left(\left[Y_i Y_j - E_\pi\left(Y_i Y_j\right)\right] Y_k Y_l\right) \right| 
            &= \left|E_\pi \left(\left[Y_i Y_j - E_\pi\left(Y_i Y_j\right)\right] Y_k \ K^{l-k} f_0(X_k)\right) \right| \\
            &\stackrel{(\star_{10})}{\leq} \left[E_\pi\left( \left[Y_i Y_j - E_\pi \left(Y_i Y_j\right)\right]^2 Y_k^2 \right)\right]^{\frac{1}{2}} \left[ E_\pi\left(K^{l-k} f_0(X_k)\right)^2 \right]^{\frac{1}{2}} \\
            &\leq 8 \sqrt{\mu_2 \mu_6} \ \lambda^{l-k} \numbereqn \label{U_5_2_term_upper_bd_1}
            \end{align*}
            and due to reversibility, 
            \begin{align*}
            \left| E_\pi \left(\left[Y_i Y_j - E_\pi\left(Y_i Y_j\right)\right] Y_k Y_l\right) \right|
            &= \left|E_\pi\left(Y_i Y_j \left[Y_k Y_l - E_\pi \left(Y_k Y_l\right) \right] \right) \right| 
            \leq 8 \sqrt{\mu_2\mu_6} \ \lambda^{j-i} \numbereqn \label{U_5_2_term_upper_bd_2}. 
            \end{align*}
            Finally, we let
            \[
            H(X_j) = E \left[ \left(Y_i Y_j - E_\pi\left(Y_i Y_j\right)\right) \mid X_j, X_k, X_l \right]
            = E \left[\left(Y_i Y_j - E_\pi\left(Y_i Y_j\right)\right) \mid X_j\right] \in L_0^2(\pi)
            \]
            with the equality being a consequence of the Markov property. Then, for $i < j < k < l$,
            \begin{align*}
            \left| E_\pi \left(\left[Y_i Y_j - E_\pi\left(Y_i Y_j\right)\right] Y_k Y_l\right) \right|
            &=\left|E_\pi \left(H(X_j) Y_k Y_l \right)\right| \\
            &= \left|E_\pi  E \left[H(X_j) Y_k Y_l\mid X_k, X_l\right]\right| \\
            &\leq E_\pi \left|K^{k-j} H(X_k) Y_k Y_l \right| \\
            &\stackrel{(\star_{11})}{\leq} \left[E_\pi \left(K^{k-j} H(X_k) \right)^2 \right]^{\frac{1}{2}} \left[E_\pi \left(Y_k^2 Y_l^2\right)\right]^{\frac{1}{2}} \\
            &\stackrel{(\star_{12})}{\leq} \lambda^{k-j} \left(E_\pi\left[H^2(X_k)\right]\right)^{\frac{1}{2}}\ \mu_4^{\frac{1}{2}} \\
            &\leq 4 \lambda^{k-j}\ \mu_4. \numbereqn \label{U_5_2_term_upper_bd_3}
            \end{align*}
            
            \noindent It follows from \eqref{U_5_2_term_upper_bd_1}, \eqref{U_5_2_term_upper_bd_2} and \eqref{U_5_2_term_upper_bd_3} 
            that \begin{align*}
            \left| E_\pi \left(\left[Y_i Y_j - E_\pi\left(Y_i Y_j\right)\right] Y_k Y_l\right) \right|
            &\leq \min \left\{\lambda^{l-k}, \lambda ^{j-i}, \lambda^{k-j} \right\} \left(4\mu_4 + 8\sqrt{\mu_2 \mu_6} \right) \\
            &\leq \lambda^{\max\{l-k, j-i, k-j\}} \left(4 \mu_4 + 8\sqrt{\mu_2 \mu_6}\right) \\
            &= \left(\lambda^{\frac{1}{3}}\right)^{3\max\{l-k,j-i,k-j\}} \left(4\mu_4 + 8\sqrt{\mu_2\mu_6}\right) \\
            &\leq \left(\lambda^{\frac{1}{3}}\right)^{l-i} \left(4\mu_4 + 8\sqrt{\mu_2\mu_6}\right).  
            \end{align*}
            Hence, 
            \begin{align*}
            \left|U_5^{(2)}\right| 
            &\leq \frac{24}{n^2} \sum_{i<j<k<l} \left(\lambda^{\frac{1}{3}}\right)^{l-i} \left(4\mu_4 + 8\sqrt{\mu_2\mu_6}\right) \\
            &=\frac{24}{n^2} \left(4\mu_4 + 8\sqrt{\mu_2\mu_6}\right) \sum_{r=3}^{n-1} (n-r) \binom{r-1}{2} \left(\lambda^{\frac{1}{3}}\right)^r \\
            &\leq \frac{24}{n}\left(4\mu_4 + 8\sqrt{\mu_2\mu_6}\right) \sum_{r=1}^\infty r^2 \left(\lambda^{\frac{1}{3}}\right)^r \\
            &=  \frac{24}{n}\left(4\mu_4 + 8\sqrt{\mu_2\mu_6}\right) \lambda^{\frac{1}{3}} \left(1 + \lambda^{\frac{1}{3}}\right) \left(1 - \lambda^{\frac{1}{3}}\right)^{-3}
            \to 0 \text{ as } n \to \infty. \numbereqn \label{U_5_2_limit} 
            \end{align*}
            Therefore, from \eqref{U_5_1_limit} and \eqref{U_5_2_limit}, it follows that
            \begin{equation} \label{U_5_limit}
            U_5 \to 3 \left(2\sum_{r=1}^\infty \gamma_r \right)^2 \text{ as } n \to \infty.
            \end{equation}
            
            \noindent Finally, combining \eqref{U_1_limit}, \eqref{U_2_limit}, \eqref{U_3_limit}, \eqref{U_4_limit} and \eqref{U_5_limit}, we get
            \begin{align*}
            \frac{1}{n^2} E_\pi\left[\left(Y_1 + Y_2 + \dots + Y_n \right)^4 \right] 
            &= U_1 + U_2 + U_3 + U_4 + U_5 \\
            &\to 3\mu_2^2 + 12 \mu_2 \sum_{r=1}^\infty \gamma_r + 3\left(2 \sum_{r=1}^\infty\gamma_r \right)^2 \\
            &= 3\left(\mu_2 + 2\sum_{r=1}^\infty \gamma_r \right)^2 = 3 \sigma_f^4 \text{ as } n\to\infty.
            \end{align*}
            This completes the proof.
        \end{proof}

        \begin{proposition} \label{prop_Y1barsqY2barsq}
            Under the setup assumed in Theorem~\ref{thm_clt_bme}, and if in addition the Markov chain is stationary, then	$E_\pi\left(b_n^2 \bar{Y}_1^2\bar{Y}_2^2\right) \to \sigma_f^4$ as $n \to \infty$.
        \end{proposition}
        \begin{proof}
            We have
            \begin{align*}
            E_\pi\left(b_n^2 \bar{Y}_1^2\bar{Y}_2^2\right)
            &= \frac{1}{b_n^2} \left[E_\pi \left(\sum_{i=1}^{b_n} \sum_{j=b_n+1}^{2b_n} Y_i^2 Y_j^2 \right) + E_\pi  \left(\sum_{i\neq i'} \sum_{j=b_n+1}^{2b_n} Y_i Y_{i'} Y_j^2 \right) \right.  \\
            &\qquad \quad +  \left. E_\pi \left(\sum_{i=1}^{b_n} \sum_{j\ne j'} Y_i^2 Y_j Y_{j'} \right) + E_\pi\left(\sum_{i\ne i'} \sum_{j\ne j'} Y_i Y_{i'}^2 Y_j Y_{j'} \right) \right] \\
            &= \mu_2^2 + \frac{1}{b_n^2} E_\pi \left[\sum_{i=1}^{b_n} \sum_{j=b_n+1}^{2b_n} Y_j^2  \left[Y_i^2 - E_\pi\left(Y_i^2 \right) \right] \right] \\
            &\quad + \frac{1}{b_n^2} 2 b_n \mu_2 \sum_{i<i'} \langle f_0, K^{i'-i} f_0 \rangle_\pi + \frac{1}{b_n^2} E_\pi \left[\sum_{i\ne i'} \sum_{j-b_n+1}^{2b_n} Y_i Y_{i'} \left[Y_j^2 - E_\pi \left(Y_j^2 \right) \right]\right] \\
            &\quad +  \frac{1}{b_n^2} 2 b_n \mu_2 \sum_{j<j'} \langle f_0, K^{j'-j} f_0 \rangle_\pi + \frac{1}{b_n^2} E_\pi \left[\sum_{i=1}^{b_n} \sum_{j\ne j'} Y_j Y_{j'} \left[Y_i^2 - E_\pi \left(Y_i^2\right) \right]\right] \\
            &\quad + \frac{4}{b_n^2} \left(\sum_{i<i'} \langle f_0, K^{i'-i} f\rangle_\pi\right) \left(\sum_{j<j'} \langle f_0, K^{j'-j} f_0 \rangle_\pi \right) \\
            &\quad + \frac{1}{b_n^2} E_\pi \left[\sum_{i\ne i'} \sum_{j\ne j'} \left[Y_i Y_{i'} - E_\pi \left(Y_i Y_{i'}\right) \right] \left[Y_j Y_{j'} - E_\pi \left(Y_jY_{j'}\right) \right]\right] \\
            &= \mu_2^2 + T_1 + \frac{1}{b_n} 2 \mu_2 \sum_{i<i'} \langle f_0, K^{i'-i} f_0 \rangle_\pi + T_2 + \frac{1}{b_n} 2  \mu_2 \sum_{j<j'} \langle f_0, K^{j'-j} f_0 \rangle_\pi \\
            & \quad  + T_3 +  \frac{4}{b_n^2} \left(\sum_{i<i'} \langle f_0, K^{i'-i} f\rangle_\pi\right) \left(\sum_{j<j'} \langle f_0, K^{j'-j} f_0 \rangle_\pi \right) + T_4, \text{ say}.
            \end{align*}		
            By analysis similar to the proof of Proposition~\ref{prop_Y1bar4}, it follows that for each $i = 1, 2, 3, 4$, $T_i \to 0$ as $n \to \infty$. Therefore, by the dominated convergence theorem, as $n \to \infty$,
            \begin{equation*}
            E_\pi \left(b_n^2 \bar{Y_1}^2\bar{Y_2}^2\right) \to \mu_2^2 + 4\mu_2\sum_{r=1}^{\infty} \langle f_0, K^r f_0\rangle_\pi + \left(2\sum_{r=1}^{\infty} \langle f_0, K^r f_0\rangle_\pi\right)^2 
            = \left(\mu_2 + \sum_{r = 1}^\infty \langle f_0, K^r f_0\rangle_\pi \right)^2 = \sigma_f^4.
            \end{equation*}
            This completes the proof.
            
        \end{proof}

        \begin{proposition} \label{prop_U_n}
            Consider the quantity $U_n$ as defined in \eqref{defn_U_n}. We have $E_\pi (U_n^2) \to 2 \sigma_f^4$ as $n \to \infty$.
        \end{proposition}
        
        \begin{proof}
            We have,
            \begin{align*}
            E_\pi(U_n^2) 
            &= E_\pi\left[b_n \bar{Y}_1^2 - b_n h(X_0) + b_n \tilde{h}(X_{b_n}) \right]^2 \\
            &= E_\pi \left(b_n^2 \bar{Y}_1^4 \right) + E_\pi\left[b_n^2 h^2(X_0) \right] + E_\pi\left[b_n^2 \tilde{h}^2(X_{b_n}) \right] - 2 E_\pi \left[b_n^2 \bar{Y}_1^2 h(X_0) \right] \\
            &\qquad -  2 E_\pi \left[b_n^2 h(X_0) \tilde{h}(X_{b_n}) \right] + 2 E_\pi\left[b_n^2 \bar{Y}_1^2 \tilde{h}(X_{b_n})\right] \\
            &= E_\pi \left(b_n^2\bar Y_1^4\right) + E_\pi\left[b_n^2 h^2(X_0) \right] + E_\pi\left[b_n^2 h^2(X_0) \right] - \left[E_\pi \left(b_n\bar Y_1^2\right) \right]^2 \\
            &\qquad  - 2 E_\pi \left[b_n^2 h^2(X_0) \right] - 2 b_n^2 \langle \tilde{h}, K^{b_n} \tilde{h} \rangle_\pi + 2 b_n^2 E_\pi\left(\bar Y_1^2\ \bar{Y}_2^2 \right) - 2 b_n^2 E_\pi \left(\bar Y_1^2\right) E_\pi\left(\bar{Y}_2^2\right) \\
            &= E_\pi \left(b_n^2 \bar Y_1^4 \right) - 3\left[E_\pi\left(b_n \bar Y_1^2\right)\right]^2 + 2 b_n^2 E_\pi \left(\bar Y_1^2\ \bar{Y}_2^2\right) - 2b_n^2 \langle \tilde{h}, K^{b_n} \tilde{h} \rangle_\pi
            \end{align*} 
            Of course, $E_\pi\left(b_n \bar Y_1^2 \right) \to \sigma_f^2$, and from Proposition~\ref{prop_Y1bar4} and \ref{prop_Y1barsqY2barsq}, it follows that as $n \to \infty$, $E_\pi (b_n^2 \bar Y_1^4) \to 3\sigma_f^4$ and $b_n^2 E_\pi(\bar Y_1^2\ \bar{Y}_2^2) \to  \sigma_f^4$ respectively. Finally, 
            \[
            \left| b_n^2 \langle \tilde{h}, K^{b_n} \tilde{h} \rangle_\pi \right| \leq b_n^2\ \lambda^{b_n} \left\|\tilde h \right\|^2_\pi \leq \lambda^{b_n}\ E_\pi\left(b_n^2 \bar Y_1^4\right) \to 0
            \]
            as $n \to \infty$. Consequently,
            \[
            E_\pi \left(U_n^2 \right) \to 2 \sigma_f^4 \text{ as } n \to \infty.
            \]
            This completes the proof.
        \end{proof}

        \begin{proposition} \label{prop_V_n}
            Consider the quantity $V_n$ as defined in \eqref{defn_V_n}. We have $E_\pi (V_n^2) \to 0$ as $n \to \infty$.
        \end{proposition}
        
        \begin{proof}
            We have
            \begin{align*}
            E_\pi[V_n^2] 
            &= E_\pi \left[b_n \tilde{g}\left(X_{b_n}\right) - b_n\tilde{h} \left(X_{b_n}\right) - E_\pi \left[b_n\tilde{g} \left(X_{b_n}\right) \mid X_0 \right]\right]^2 \\
            &= E_\pi\left[b_n \left( \left(I-K^{b_n}\right)^{-1} - I\right)\tilde{h}(X_{b_n})-K^{b_n}\tilde{g}(X_0)\right]^2\\
            &\leq 2E_\pi\left[b_n\left( \left(I-K^{b_n}\right)^{-1} - I\right)\tilde{h}(X_{b_n})\right]^2+2E_\pi\left[b_nK^{b_n}\tilde{g}(X_0)\right]^2 \\
            &\leq 2 \left\| \left(I-K^{b_n}\right)^{-1} - I\right\|^2 \left\|b_n\tilde{h}\right\|_\pi^2 + 2 \left\|K^{b_n}\right\|^2 \left\|b_n\tilde{g}\right\|_\pi^2 \\
            &\leq 2 \left\| \left(I-K^{b_n}\right)^{-1} - I\right\|^2 \left\|b_n\tilde{h}\right\|_\pi^2 + 2 \left\|K^{b_n}\right\|^2 \left\| \left(I-K^{b_n}\right)^{-1} \right\|^2 \left\|b_n \tilde{h}\right\|_\pi^2 \\
            &\leq 2\frac{\lambda^{2b_n}}{\left(1-\lambda^{b_n}\right)^2} E_\pi\left(b_n^2\bar{Y}_1^4\right) + 2\frac{\lambda^{2b_n}}{\left(1-\lambda^{b_n}\right)^2} E_\pi\left(b_n^2\bar{Y}_1^4\right) \\
            &= 4\frac{\lambda^{2b_n}}{\left(1-\lambda^{b_n}\right)^2} E_\pi\left(b_n^2\bar{Y}_1^4\right) 
            \end{align*}
            where $\lambda = \|K\| \in (0, 1)$. From Proposition~\ref{prop_Y1bar4} it follows that $E_\pi\left(b_n^2\bar{Y}_1^4\right) \to 3 \sigma_f^4$. Hence, $E_\pi(V_n^2) \to 0$ as $n \to \infty$. This completes the proof.
        \end{proof}

    \end{appendix}
\end{document}